\documentclass{elsarticle}

\pdfoutput=1

\usepackage{amsmath,amssymb,verbatim,url}
\usepackage{lineno, hyperref}

\def\N{\mathbb{N}}
\def\Z{\mathbb{Z}}
\def\Q{\mathbb{Q}}

\def\K{\mathbb{K}}

\newcommand{\myatop}[2]{\genfrac{}{}{0pt}{}{#1}{#2}}

\DeclareMathOperator{\gcrd}{gcrd}
\DeclareMathOperator{\ord}{ord}

\newcommand {\cB} {{\cal B}}
\newcommand {\cP} {{\cal P}}
\newcommand {\cC} {{\cal C}}
\newcommand {\cL} {{\cal L}_\cB}

\newcommand {\cRc} {{\cal R}_\cC}
\newcommand {\cR}[1][\cB] {{\cal R}_{#1}}

\newcommand {\cE} {{\cal E}}

\def\KK{\K[[\cB]]}

\usepackage{amsthm}
\newtheorem{theorem}{Theorem}
\newtheorem{corollary}[theorem]{Corollary}
\newtheorem{lemma}[theorem]{Lemma}
\newtheorem{proposition}[theorem]{Proposition}

\theoremstyle{remark}
\newtheorem{remark}[theorem]{Remark}

\theoremstyle{definition}
\newtheorem{definition}[theorem]{Definition}
\newtheorem{example}[theorem]{Example}
\newtheorem{notation}[theorem]{Notation}

\newcommand{\bc}{\ensuremath{\mathbf{c}}}

\newcommand{\code}[1]{\texttt{#1}}
\newcommand{\packname}{\code{pseries\_basis}}
\newcommand{\packurl}{{\url{https://github.com/Antonio-JP/pseries\_basis}}}
\newcommand{\packbinder}{{\url{https://mybinder.org/v2/gh/Antonio-JP/pseries_basis/master?labpath=notebooks\%2Fpaper_examples.ipynb}}}
\newcommand{\packappb}{{\url{https://mybinder.org/v2/gh/Antonio-JP/pseries_basis/master?labpath=notebooks\%2Fpaper_examples_appB.ipynb}}}
\newcommand{\packdoc}{{\url{https://antonio-jp.github.io/pseries_basis/}}}

\usepackage{xcolor}
\usepackage{listings}
\lstdefinelanguage{Sage}[]{Python}
{morekeywords={False,sage,True},sensitive=true}
\definecolor{dblackcolor}{rgb}{0.0,0.0,0.0}
\definecolor{dbluecolor}{rgb}{0.01,0.02,0.7}
\definecolor{dgreencolor}{rgb}{0.2,0.4,0.0}
\definecolor{dgraycolor}{rgb}{0.30,0.3,0.30}

\begin{document}

\begin{frontmatter}

\title{The Factorial-Basis Method\\ for Finding Definite-Sum Solutions of\\ Linear Recurrences With Polynomial Coefficients}

\author[a]{Antonio Jim\'enez-Pastor}
\author[b,c]{Marko Petkov\v sek\corref{ca}}
\cortext[ca]{Corresponding author}
\ead{marko.petkovsek@fmf.uni-lj.si}

\address[a]{LIX, CNRS, \'{E}cole Polytechnique, Institute Polytechnique de Paris, France}
\address[b]{Faculty of Mathematics and Physics, University of Ljubljana, Slovenia}
\address[c]{Institute of Mathematics, Physics and Mechanics, Ljubljana, Slovenia}

\begin{abstract}
The problem of finding a nonzero solution of a linear recurrence $Ly=0$ with polynomial coefficients where $y$ has the form of a definite hypergeometric sum, related to the Inverse Creative Telescoping Problem of \cite{ChKa}[Sec.~8], has now been open for three decades. Here we present an algorithm (implemented in a SageMath package) which, given such a recurrence and a \emph{quasi-triangular}, \emph{shift-compatible factorial basis} $\cB = \langle P_k(n) \rangle_{k=0}^{\infty}$ of the polynomial space $\K[n]$ over a field $\K$ of characteristic zero, computes a recurrence satisfied by the coefficient sequence $c = \langle c_k \rangle_{k=0}^{\infty}$ of the solution $y_n = \sum_{k=0}^\infty c_k P_k(n)$ (where, thanks to the quasi-triangularity of $\cB$, the sum on the right terminates for each $n\in\N$). More generally, if $\cB$ is $m$-sieved for some $m\in\N$, our algorithm computes a system of $m$ recurrences satisfied by the $m$-sections of the coefficient sequence $c$. If an explicit nonzero solution of this system can be found, we obtain an explicit nonzero solution of $Ly=0$.
\end{abstract}


\end{frontmatter}


\section{Introduction}

By definition, a \emph{P-recursive} (or: \emph{holonomic}) sequence over a field $\K$ of characteristic 0 is given by a homogeneous linear recurrence with polynomial coefficients, together with suitable initial conditions. Often one wishes to find \textit{explicit representations} of P-recursive sequences, so algorithms have been devised to find solutions of such recurrences within a given class of explicitly representable sequences. Some well-known examples of this kind are the algorithms for finding polynomial \cite{A89poly}, rational \cite{A89rat, unden}, hypergeometric \cite{hyper, MvH99}, d'Alembertian \cite{AP94}, and Liouvillian solutions \cite{HS}, as well as hypergeometric solutions in the setting of $\Pi\Sigma^*$-fields~\cite{ABPS21}. These classes do not exhaust explicitly representable P-recursive sequences; for instance, every definite hypergeometric sum on which Zeilberger's Creative Telescoping algorithm \cite{Zeil90, Zeil91} succeeds is a P-recursive sequence, but such sequences are typically not Liouvillian. Hence it makes sense to consider the \textit{Inverse Creative Telescoping Problem} (ICTP): given a homogeneous linear recurrence with polynomial coefficients with no Liouvillian solutions, find its solutions in the form of \emph{definite sums} of a given type. Problems of this kind have been posed in \cite[p.~84, item 2]{CMUThesis}, and again in \cite[Sec.~8]{ChKa}. 

Here we make a small but important step towards solving ICTP. Let $L y = 0$ be the equation we wish to solve, where $L$ is a linear recurrence operator with polynomial coefficients. We provide an algorithm which, given $L$ and a sequence $\cB=\langle P_k(n)\rangle_{k=0}^\infty$ of polynomials in $n$ which is a \emph{quasi-triangular, shift-compatible factorial basis} of the polynomial space $\K[n]$ (see Definitions~\ref{def:factorial},~\ref{def:compat} and~\ref{qt}),
returns a linear recurrence operator $L'$ such that for any sequence $y\in\K^\N$ of the form 
\begin{equation}
\label{defsumsol}
y_n\ =\ \sum_{k=0}^\infty c_k P_k(n)\ \text{ for all } n\in \N
\end{equation}
for some $c\in\K^\N$, we have $L y = 0$ if and only if $L' c = 0$. So, if we can solve the latter equation for the unknown sequence $c$ (by, say, one of the algorithms mentioned in the preceding paragraph, or by using our \emph{factorial-basis method} (FBM) recursively), $y$ in  (\ref{defsumsol}) will be an explicit definite-sum solution of $Ly = 0$. 
 
We point out that FBM can be viewed as a kind of a discrete converse of the method of generating functions (GFM) for solving recurrences of the form $Ly = 0$. As is well known, GFM produces a \emph{differential} equation $L'f = 0$ satisfied by the (ordinary) generating function $f(x) = \sum_{n=0}^\infty y_n x^n$, which we then solve (if we can) and read off the coefficient sequence $\langle y_n\rangle_{n=0}^\infty$ from the obtained solution. On the other hand, FBM assumes that the unknown sequence $y$ is of the form (\ref{defsumsol}) (with the basis $\cB$ given as part of the input, playing the role of the power basis $\langle x^n\rangle_{n=0}^\infty$ in GFM),
and produces a \emph{recurrence} equation $L'c = 0$  which we then solve for $c$ (if we can), and obtain the solution $y$ from (\ref{defsumsol}).  
Like other solution methods, FBM can also be used for \emph{factoring linear recurrence operators} (cf.~Example~\ref{ex:main}), as well as for \emph{deriving summation identities} (cf.~Examples~\ref{bigex},~\ref{ex:main}) when some other form of the solution is already known.
 %
 
There is also some remote similarity between our problem where, given a recurrence operator $L$ 
and a basis $\langle P_k(n) \rangle_{k=0}^{\infty}$ of the space of polynomials $\K[n]$, we seek sequences $c\in\K^\N$ such that $y$ in (\ref{defsumsol}) satisfies $Ly = 0$, and the classical \emph{Fredholm} or \emph{Volterra integral equations of the first kind}, as well as the  \emph{Stieltjes moment problem} (cf.~\cite{Mnat8_13}) having the form
\begin{align*}
g(n)\ &=\ \int_0^\infty K(n,t) f(t) dt \qquad \text{ (Fredholm)}\\
g(n)\ &=\ \int_0^n K(n,t) f(t) dt \qquad \text{ \ (Volterra)}\\
m_n\ &=\ \int_0^\infty t^n f(t) dt \qquad\quad \text{ \ \ \ (Stieltjes)}
\end{align*}
where the left-hand sides $g(n)$ resp.\ $\langle m_n\rangle_{n=0}^\infty$, as well as the \emph{kernels} $K(n,t)$ resp.\ $t^n$ are given, and one seeks the unknown function $f(t)$.
Apart from our problem being ``discrete'' while the above three are ``continuous'', the main difference between them lies in the fact that our sequence $y$ in (\ref{defsumsol}) is given \emph{recursively}, and our goal is to find its explicit representation in terms of the unknown $c$, while in the other three problems the left-hand sides are presumably given explicitly, and finding $f(t)$ (corresponding to our $c$) is the final goal. Nevertheless, we will occasionally write our polynomial basis element $P_k(n)$ as $K(n,k)$, and call it the \emph{kernel} of (\ref{defsumsol}).

Note that recently, Imamoglu and van Hoeij \cite{MvH17} have solved 
the important related problem of finding definite-sum solutions of second-order linear \emph{differential} equations with rational-function coefficients. Their algorithms (very effective in practice, but called ``heuristic'' by the authors as they haven't been fully proven yet) find solutions in the form $A \cdot\, _2F_1(a_1,a_2;b_1;f)$ or in the form 
\[
A \cdot \left(r_0 \cdot {_2F_1}(a_1,a_2;b_1;f)\ +\ r_1 \cdot {_2F'_1}(a_1,a_2;b_1;f)\right)
\]
where $A$ has algebraic logarithmic derivative, and $f, r_0, r_1$ are algebraic. 

 The contents of the rest of the paper are as follows:
 In Section \ref{sec:FPS}, we define \emph{factorial bases} of the polynomial algebra $\K[x]$, and the notion of their \emph{compatibility with endomorphisms} of $\K[x]$. Following \cite{APR}, to each factorial basis $\cal B$ we assign the algebra $\K[[\cal B]]$ of \emph{formal polynomial series} as a generalization of the algebra $\K[[x]]$ of formal power series, with the basis element $P_k(x) \in {\cal B}$ in the former algebra playing the role of $x^k$ in the latter. We 
extend the action of an endomorphism $L$ of $\K[x]$ to $\K[[\cal B]]$ in a natural way, then assign to $L$ its \emph{associated operator} $L'=\cR L$ acting on sequences in such a way that $L\left(\sum_{k=0}^\infty c_k P_k(x)\right) = 0$ iff $\sum_{k=0}^\infty (L' c_k)\, P_k(x) = 0$. This enables us to solve the equation $Ly=0$ for $y\in \K[[\cal B]]$ by solving the (perhaps simpler) equation $L' c = 0$ for $c \in \K^\N$.

In order for FBM to be useful,
we need our formal polynomial series $y$ to have a definite value $y_n \in \K$ for every $n\in \N$. For example, when $P_k(n) = \binom{n}{k}$, this is true since 
the series $y$ in (\ref{defsumsol}) is terminating (in fact, $y$ is the classical \emph{binomial transform} of $c$).  With this example in mind, in Section \ref{sec:BinTransf} we define \emph{quasi-triangular bases} which are factorial shift-compatible bases with the properties that for a fixed $n$, we have $P_k(n) = 0$ for all $k$ large enough compared with $n$, and that for each $a\in\K^\N$ there is some $b\in \K^\N$ such that $a_n = \sum_{k=0}^\infty b_k P_k(n)$. 

Section \ref{sec:ProdBases} defines generalized binomial-coefficient bases, and provides a mechanism for creating many new compatible bases by taking \emph{products} of the already constructed ones. In bases that are products of $m > 1$ factors, the coefficients $\alpha_{k,i}$ expressing the actions of the operators on the basis elements are quite complicated conditional expressions depending on the residue class of $k \bmod m$. Therefore in Section \ref{sec:Sieved} we extend our approach to the so-called \textbf{\emph{sieved polynomial bases}} where the definition of the $k$-th basis element $P_k(x)$ depends on the residue class of $k \bmod m$. To facilitate the computations, we do not attempt to compute the associated operator $L' = \cR L$ directly but instead represent it by a matrix of operators $[\cR L] = [L_{r,j}]_{r,j=0}^{m-1}$ where the operator $L_{r,j}$ expresses the contribution of the $j$-th $m$-section of the coefficient sequence of $y$ to the $r$-th $m$-section of the coefficient sequence of $Ly$. Section \ref{sec:ex} presents several nontrivial applications of the developed theory and algorithms, such as the explicit solution of a recurrence equation of order 7 (which leads to complete factorization of the corresponding operator -- see Example~\ref{ex:main}), and construction of explicit solutions of the so-called \emph{Ap\'ery recurrences} for $\zeta(2)$ and $\zeta(3)$ (\ref{exm:apery2}, and \ref{exm:apery3}). 

In Section~\ref{sec:gen} we introduce a particular instance of sieved polynomial bases called
\textbf{\emph{shuffled polynomial bases}}, which are defined as some specific interlacing or shuffling 
of basic sieved polynomial bases. This can be seen as a straightforward generalization of the product bases 
defined in Section~\ref{sec:ProdBases}. This section includes a fully constructive way of extending the 
compatibilities of different operators.


\begin{notation} $\N = \{0,1,2,\ldots\}$ denotes the set of nonnegative integers, $\K$ a field of characteristic zero, $\K^{\N}$ the set of all sequences with terms from $\K$, $\K[x]$ the $\K$-algebra of univariate polynomials over $\K$, and ${\cal L}_{\K[x]}$ the $\K$-algebra of linear operators $L: \K[x] \rightarrow \K[x]$. 
\end{notation}

\begin{definition}
\label{def1}
 Let  $m \in \N \setminus \{0\}$ and $j \in \{0,1,\ldots,m-1\}$.
\begin{itemize}
\item A sequence $c \in \K^{\N}$ is called the \emph{$j$-th $m$-section} of a sequence $a \in \K^{\N}$ if $c_k = a_{mk+j}$ for all $k \in \N$. We say that $c$ is obtained from $a$ by \emph{multisection}, and denote it by $s_j^m a$.
\item A sequence $c =\Lambda(a^{(0)},a^{(1)}, \ldots, a^{(m-1)}) \in \K^{\N}$ is called
 the \emph{interlacing} of se\-quen\-ces $a^{(0)},a^{(1)}, \ldots, a^{(m-1)} \in \K^{\N}$ if  $c_k = a_q^{(r)}$ where $k = qm + r$ with $q \in \N$ and $r \in \{0,1,\dots,m-1\}$ for all $k \in \N$. 
\end{itemize}
\end{definition}

\section{Formal polynomial series}
\label{sec:FPS}
 
The power-series method is a time-honored approach to solving differential equations by reducing them to  recurrences satisfied by the coefficient sequences of their power series solutions. In \cite{APR} it was shown how, by generalizing the notion of formal \emph{power series} to formal \emph{polynomial series}, one can use this method to find solutions of other linear operator equations such as $q$-difference equations, and recurrence equations themselves, which interest us here. In this section we summarize some relevant definitions, examples and results from \cite{APR}.

\begin{definition}\label{def:factorial}
We call a sequence of polynomials $\cB = \langle P_k(x) \rangle_{k=0}^{\infty}$ from $\K[x]$ a \emph{factorial basis} of $\K[x]$, if for all $k \in \N$:
\begin{enumerate}
\item[\bf P1.] $\deg P_k(x) = k$, 
\item[\bf P2.] $P_k(x) \,|\, P_{k+1}(x)$ in $\K[x]$. 
\end{enumerate}
\end{definition}
\noindent
Note that due to property \textbf{P1}, any factorial basis of $\K[x]$ is a basis of $\K[x]$ as a vector space over $\K$.
\begin{notation}
Denote by $\cP = \left\langle x^k\right\rangle_{k=0}^{\infty}$ the {\em power basis}, and by 
$\cC = \left\langle {x\choose k}\right\rangle_{k=0}^{\infty}$ the {\em binomial-coefficient basis} of $\K[x]$, respectively. 
\end{notation}

\begin{example}
Clearly, both $\cP$ and $\cC$  are factorial bases of $\K[x]$.
\end{example}

\begin{proposition}
${\cal B} = \langle P_k(x)\rangle_{k=0}^\infty$ is a factorial basis iff there are a \emph{root sequence} $\rho = \langle \rho_1,\rho_2,\rho_3,\ldots\rangle \in \K^{\N\setminus \{0\}}$ and a sequence $\langle c_0,c_1,c_2,\ldots\rangle \in {(\K^*)}^{\N}$ such that
\begin{align}
\label{roots}
P_k(x)\ =\ c_k (x-\rho_1)(x-\rho_2)\cdots (x-\rho_k) \ \ \text{for all } k \in \N.
\end{align}
\end{proposition}

\begin{proof}
If (\ref{roots}) holds then $\langle P_k(x)\rangle_{k=0}^\infty$ clearly satisfies \textbf{P1} and \textbf{P2}.

\smallskip
Conversely, if $\langle P_k(x)\rangle_{k=0}^\infty$ satisfies \textbf{P1} and \textbf{P2} then for each $k \in \N$ there are $u_k \in \K^*$ and $v_k \in \K$ such that $P_{k+1}(x) =(u_k x-v_k) P_k(x) = u_k P_k(x) (x - v_k/u_k)$. Let $c_0 := P_0(x) \in K^*$. By induction on $k$, we see that each $P_k(x)$ is of the form (\ref{roots}) with $c_k = c_0\prod_{j=0}^{k-1} u_j$ and $\rho_k = v_k/u_k$ for all $k \in \N\setminus \{0\}$. 
\end{proof}

\begin{example}
 The root sequence of the power basis $\cal P$ is $\rho = \langle 0,0,0,0,\dots \rangle$, and $c_k=1$ for all $k \in \N$. The root sequence of the binomial-coefficient basis $\cal C$ is $\rho = \langle 0,1,2,3,\dots \rangle$, and $c_k=1/k!$\, for all $k \in \N$. 
\end{example}

Note that in the umbral calculus, factorial bases with $c_k=1$ for all $k \in \N$ are known as \emph{sequences of polynomials with persistent roots} (cf.~\cite{dBL}).

\begin{notation}
Denote by $D$, $E$, $Q$, $X \in {\cal L}_{\K[x]}$ the \emph{differentiation}, \emph{shift}, \emph{q-shift}, and \emph{multiplication-by-the-independent-variable operators}, respectively, acting on polynomials $p \in \K[x]$ by
\[
\begin{array}{lll}
D p(x) &=& p'(x), \\
E p(x) &=& p(x+1), \\
Q p(x) &=& p(qx), \\
X p(x) &=& x p(x)
\end{array}
\]
where $q \in \K^*$ is not a root of unity, so that the operators $\{1, Q, Q^2,\ldots\}$ are linearly independent.
\end{notation}

\begin{definition}
\label{def:compat}

A factorial basis $\cB$ of $\K[x]$ and an operator $L \in {\cal L}_{\K[x]}$ are \emph{compatible} with each other if there are $A, B \in \N$ such that, for all $k \in \N$, there are $\alpha_{k,i} \in \K$ with $-A \le i \le B$, such that
\begin{equation}
\label{main}
LP_k(x) = \sum_{i=-A}^B \alpha_{k,i}\;P_{k+i}(x)\ \ { for\ all\ } k \in \N,
\end{equation}
with $P_j(x) = 0$ when $j < 0$. To assert that (\ref{main}) holds for specific $A, B \in \N$, we will say 
that $\cB$ is \emph{$(A,B)$-compatible} with $L$.
\end{definition}

\begin{example}
\label{basis}
\leavevmode
\begin{itemize}
\item $(x^k)'=k x^{k-1}$, so $\cP$ is (1,0)-compatible with $D$ (simply, take $\alpha_{k,-1}=k$, $\alpha_{k,0} = 0$),
\item $\binom{x+1}{k} = \binom{x}{k-1}+\binom{x}{k}$, so $\cC$ is (1,0)-compatible with $E$ (take $\alpha_{k,-1} = 1$, $\alpha_{k,0} = 1$),
\item $(q x)^k = q^k x^k$, so $\cP$ is (0,0)-compatible with $Q$ (take $\alpha_{k,0}=q^k$),
\item the basis $\cB = \langle \binom{x+k}{k}\rangle_{k=0}^\infty$ is \emph{not} compatible   with $E$, since the well-known identity $\binom{x+k+1}{k} = \sum_{j=0}^k \binom{x+j}{j}$ implies that $E P_k(x) = \sum_{j=0}^k P_j(x)$, where the upper bound depends on $k$.
\end{itemize}
\end{example}

\begin{proposition}
\label{compat}
A factorial basis $\cB$ of $\K[x]$ is $(A,B)$-compatible with $L\in {\cal L}_{\K[x]}$ if and only if 
\begin{enumerate}
\item[\bf C1.] $\deg L P_k(x) \le k + B$ for all $k \ge 0$,
\item[\bf C2.] $P_{k-A}(x) \,|\, L P_k(x)$\ for all $k  \ge A$.
\end{enumerate}
\end{proposition}

\begin{proof}
Necessity of these two conditions  is obvious. For sufficiency, let 
\[
L P_k(x)\ = \sum_{j=0}^{\deg L P_k(x)} \lambda_{j,k}P_j(x)
\]
be the expansion of $L P_k(x)$ w.r.t.\ $\cB$. By \textbf{C1}, we can replace the upper summation bound by $ k+B$. Rewriting the resulting equation as
\[
LP_k(x)\ - \sum_{j=k-A}^{k+B} \lambda_{j,k} P_j(x)\ =  \sum_{j=0}^{k-A-1} \lambda_{j,k}P_j(x),
\]
we see  by \textbf{C2} and \textbf{P2} that $P_{k-A}(x)$ divides the left side, while the right side is of degree less than $k-A = \deg P_{k-A}(x)$. Hence both sides vanish, and so
\[
LP_k(x)\ = \sum_{j=k-A}^{k+B} \lambda_{j,k} P_j(x)\ = \sum_{i=-A}^{B} \lambda_{k+i,k} P_{k+i}(x)\ = \sum_{i=-A}^{B} \alpha_{k,i} P_{k+i}(x)
\]
where $ \alpha_{k,i} := \lambda_{k+i,k}$. This proves $(A,B)$-compatibility of $\cB$ with $L$. 
\end{proof}

\begin{corollary}
\label{cor_compat}
Every factorial basis is $(0,1)$-compatible with $X$.
\end{corollary}
 
\begin{proof}
Since $\deg x P_k(x) = k+1$ and $P_k(x) \,|\, x P_k(x)$ for all $k \ge 0$, this follows from Proposition \ref{compat}. 
\end{proof}

\begin{proposition}
\label{B=0_for_E}
If a factorial basis $\cB$ is $(A,B)$-compatible with $E$, then $\cB$ is also $(A,0)$-compatible with $E$.
\end{proposition}

\begin{proof}
Since the shift operator $E$ preserves polynomial degrees, the coefficients of $P_j(x)$ with $j > k$ in the expansion of $E P_k(x)$ w.r.t.\ $\cB$ all vanish. 
\end{proof}

\begin{proposition}
\label{compatroots}
A factorial basis $\cB$ of $\K[x]$ having the root sequence $\rho$ 
is $(A,0)$-compatible with the shift operator $E$ if and only if for all $k \in \N$ the following inclusion of multisets is valid:
\begin{align}
\label{rhotest}
[\rho_1 + 1, \rho_2 + 1,  \dots, \rho_k+ 1]\ \subseteq\ [\rho_1, \rho_2, \dots, \rho_{k+A}].
\end{align}
\end{proposition}

\begin{proof}
We use Proposition \ref{compat} with $L=E$ and $B=0$. Since for every factorial basis we have $\deg E P_k(x) = \deg P_k(x) = k$, condition \textbf{C1} is always satisfied. Hence $(A,0)$-compatibility of $\cB$ with $E$ is equivalent to condition \textbf{C2} which requires that $P_{k-A}(x)$ divides $P_k(x+1)$ for all $k \ge A$. In terms of $\rho$ this is equivalent to
$
[\rho_1, \rho_2,  \dots, \rho_{k-A}]\ \subseteq\ [\rho_1 - 1, \rho_2 - 1, \dots,\rho_k - 1]
$
for all $k \ge A$, or
\begin{align*}
[\rho_1, \rho_2,  \dots, \rho_{k}]\ \subseteq\ [\rho_1 - 1, \rho_2 - 1, \dots,\rho_{k+A} - 1]
\end{align*}
for all $k \ge 0$.
By adding 1 to all the terms on both sides, this turns into (\ref{rhotest}). 
\end{proof}

\begin{example}
   
   \begin{itemize}
   \item Since $[\overbrace{1,1,\dots,1}^k]\ \not\subseteq\ [\overbrace{0,0,\dots,0}^{k+A}]$ for all $k \ge 1$ and $A\ge 0$, Proposition \ref{compatroots} implies that $\cal P$ is not compatible with $E$.
   \item Since $[1,2,3,\dots,k]\ \subseteq\ [0,1,2,\dots,k-1,k]$ for all $k \ge 0$, Proposition \ref{compatroots} implies that $\cal C$ is $(1,0)$-compatible with $E$.
   \end{itemize}
\end{example}

Compatibility of a factorial basis with the differentiation operator $D$ can also be characterized
in terms of its root sequence as shown in the next result.

\begin{proposition}
   \label{compat_der_roots}
   Let $\cB$ be a factorial basis with the root sequence $\rho$. Then $\cB$ is $(p,0)$-compatible with $D$ if 
   and only if, for all $n \in \N$ we have
   \[\{\rho_1,\ldots, \rho_n\} \subset \{\rho_{n+1},\ldots, \rho_{n+p}\}.\]
\end{proposition}
\begin{proof}
   Let $g_n(x) = \gcd(P_n(x), P_n'(x))$. It is well known that, if we write $P_n(x) = A_n(x)g_n(x)$ and 
   $P_n'(x) = B_n(x)g_n(x)$, we have that $A_n(x)$ is a polynomial with simple roots and such that $A_n(\rho_m) = 0$
   for all $m \leq n$ (i.e., it contains all the different roots up to and including $\rho_n$).

   By Proposition~\ref{compat}, condition \textbf{C2}, we have that $D$ is $(p,0)$-compatible with $\cB$
   if and only if $P_{n}(x)$ divides $P_{n+p}'(x)$ for all $n\geq 0$. Write 
   \[P_{n+p}(x) = P_{n}(x)(x-\rho_{n+1})\cdots (x-\rho_{n+p}).\]
   Using the product rule, we have that $P_{n}(x)$ divides $P_{n+p}'(x)$ if and only if 
   \begin{equation}\label{equ:cdr:1}P_{n}(x)\ |\ P_{n}'(x) (x-\rho_{n+1})\cdots (x-\rho_{n+p}).\end{equation}
   Using the definition of $A_n(x)$ and $B_n(x)$ described above, we have that~\eqref{equ:cdr:1} holds if and only if:
   \[A_{n}(x)\ |\ B_{n}(x) (x-\rho_{n+1})\cdots (x-\rho_{n+p}),\]
   and using the fact that $\gcd(A_{n}(x), B_{n}(x)) = 1$, this is equivalent to:
   \[A_{n}(x)\ |\ (x-\rho_{n+1})\cdots (x-\rho_{n+p}).\]
   By the definition of $A_n(x)$, this is equivalent to
   \[\{\rho_1,\ldots, \rho_n\} \subset \{\rho_{n+1},\ldots, \rho_{n+p}\}.\] 
\end{proof}

\begin{corollary}
   \label{cor:compat_der_roots:1}
   Let $\cB$ be a factorial basis with the root sequence $\rho$. If $\cB$ is $(p,0)$-compatible with $D$, then
   the number of distinct roots in $\rho$ is at most $p$.
\end{corollary}
\begin{proof}
   Using Proposition~\ref{compat_der_roots}, if $\cB$ is $(p,0)$-compatible, then the inclusion of roots
   implies that the distinct roots of $P_n(x)$ are at most $p$. 
\end{proof}

\begin{corollary}
   \label{cor:compat_der_roots:2}
   Let $\cB$ be a factorial basis whose root sequence $\rho$ is periodic with period length $p$. 
   Then $\cB$ is $(p,0)$-compatible with $D$. Conversely, if $\cB$ is $(p,0)$-compatible with 
   $D$ and $\cB$ has $p$ different roots, then there is $n_0 \in \N$ such that 
   $\rho_{n+p} = \rho_n$ for all $n \geq n_0$.
\end{corollary}
\begin{proof}
   If the root sequence of $\cB$ is periodic with period length $p$, then $\rho_{n+p} = \rho_n$, so in
   particular, 
   \[\{\rho_1,\ldots, \rho_n\} \subseteq \{\rho_1,\ldots, \rho_p\} = \{\rho_{n+1},\ldots,\rho_{n+p}\},\]
   hence, by Proposition~\ref{compat_der_roots}, $\cB$ is $(p,0)$-compatible with $D$.
   
   On the other hand, assume that $\rho$ has $p$ different elements and let $n_0$ be the least positive integer
   such that $|\{\rho_1,\ldots,\rho_{n_0}\}| = p$. Since $\cB$ is $(p,0)$-compatible with $D$, we have that,
   for $n \geq n_0$, all the elements $\rho_{n+1},\ldots,\rho_{n+p}$ are different. We can see that from this point on 
   $\rho_n = \rho_{n+p}$ for all $n \geq n_0$. 
\end{proof}

\begin{example}
   
   \begin{itemize}
      \item Since $\cP$ has as root sequence $\langle 0, 0, 0, \ldots\rangle$, which is periodic with period length $1$, it is $(1,0)$-compatible with $D$ by Corollary~\ref{cor:compat_der_roots:2}.
      \item Since the root sequence of $\cC$ contains infinitely many different elements, by Corollary~\ref{cor:compat_der_roots:1}, $\cC$ is not compatible with $D$.
   \end{itemize}
\end{example}

Let $\cB = \langle P_k(x) \rangle_{k=0}^{\infty}$ be a factorial basis, and let $\ell_k : \K[x] \rightarrow \K$ for $k \in \N$ be linear functionals such that $\ell_k(P_m(x))=\delta_{k,m}$ for all $k, m \in \N$ (i.e.,  $\ell_k(p(x))$ is the coefficient of $P_k(x)$ in the expansion of $p(x) \in \K[x]$ w.r.t.\ $\cB$). Property {\bf P2} implies that $\ell_k(P_j(x) P_m(x)) = 0$ when $k < \max \{j,m\}$, hence $\K[x]$ naturally embeds into the \textbf{\textit{algebra $\KK$ of formal polynomial series}} of the form 
\begin{equation}
\label{series}
y(x) = \sum_{k=0}^{\infty} c_k P_k(x) \qquad (c_k\in \K),
\end{equation}
with multiplication defined by 
\[
\left(\sum_{k=0}^{\infty} c_k P_k(x)\right)
\left(\sum_{k=0}^{\infty} d_k P_k(x)\right)
\ =\ \sum_{k=0}^{\infty} e_k P_k(x),
\]
\[
e_k\ = \sum_{\max\{i,j\} \le k \le i+j} c_i d_j \,\ell_k(P_i(x) P_j(x)).
\]

\begin{notation}
\label{notL_B}
\begin{itemize}
\item \label{def4} For any factorial basis $\cB$ of $\K[x]$,  let $\cL$ denote the set of all operators $L \in {\cal L}_{\K[x]}$ such that  $\cB$ is compatible with $L$.
\item Let $\cE$ denote the $\K$-algebra of recurrence operators of the form $L' = \sum_{i=-s}^r c_i(k) S^i$ with $r,s\in\N$ and $c_i: \Z \rightarrow \K$ for $-s \le i \le r$, acting on the $\K$-algebra of all two-way infinite sequences $c \in \K^\Z$ by $S^i c(k) = c(k+i)$ for all $k,i \in \Z$.
\item For any factorial basis $\cB$ of $\K[x]$, let $\sigma_\cB$ denote the map $\K[[\cB]] \to \K^{\Z}$ assigning to $y(x) = \sum_{k=0}^\infty c_k P_k(x) \in \K[[\cB]]$ its coefficient sequence $c = \langle c_k\rangle_{k\in\N}$ extended to $\langle c_k\rangle_{k\in\Z}$ by taking $c_k = 0$ whenever $k<0$. We will omit the subscript $\cB$ when it is clear from the context. 
\end{itemize}
\end{notation}

\begin{definition}
\label{Lext}
Let $\cB = \langle P_k(x) \rangle_{k=0}^{\infty}$ be a factorial basis of $\K[x]$, $(A,B)$-compatible with $L \in {\cal L}_{\K[x]}$. Extend $L$ to an operator acting on $\KK$ by setting
\begin{eqnarray}
L \sum_{k=0}^{\infty} c_k P_k(x)  &:=&   \sum_{k=0}^{\infty} c_k L P_k(x) 
\ =\ \sum_{k=0}^{\infty} c_k \sum_{i=-A}^B \alpha_{k,i} P_{k+i}(x) \nonumber \\
&=& \sum_{k=0}^{\infty}\left( \sum_{i=-B}^A \alpha_{k+i,-i} c_{k+i}\right) P_k(x) \label{Ly}
\ =\ \sum_{k=0}^{\infty}(\cR L)c_k\, P_k(x)\qquad
\end{eqnarray}
where
\begin{equation}
\label{RL}
\cR L\ := \sum_{i=-B}^A \alpha_{k+i,-i} S^i\ \in\ \cE
\end{equation}
is the operator, \emph{associated to} $L$ in basis $\cB$,
with $A, B, \alpha_{k,i}$ as in $(\ref{main})$, and $P_k(x) = c_{k-i} = 0$ whenever $k < 0$ or $i > k$. 
\end{definition}

\begin{theorem}
\label{RBL}
Let  $\cB$, $L$ and $\cR L$ be as in Definition \ref{Lext}, and let $y, p \in \KK$. 
\begin{enumerate}
\item $\sigma_\cB(Ly)\ =\ (\cR L) \sigma_\cB y$,
\item $Ly=p\ \Longleftrightarrow\ (\cR L)\sigma_\cB y = \sigma_\cB p$.
\end{enumerate}
\end{theorem}

\begin{proof}
Write $y(x) = \sum_{k=0}^{\infty} c_k P_k(x)$. Then $\sigma_\cB y = c$ and $\eqref{Ly}$ imply that $\sigma_\cB(Ly)\ =\ (\cR L) c$, hence $\sigma_\cB(Ly)\ =\ (\cR L) \sigma_\cB y$, proving item 1. Furthermore,
\begin{align*}
L y &= p\ \ \ \Longleftrightarrow\ \ \  \sigma_\cB(L y)\, =\, \sigma_\cB p\ \ \ \Longleftrightarrow\ \ \ (\cR L)\sigma_\cB y\, =\, \sigma_\cB p,  
\end{align*}
proving item 2. 
\end{proof}

\begin{example}
\label{subst}
Using \eqref{RL} we read off from the $\alpha_{k,i}$ given in Example \ref{basis} that
\begin{eqnarray*}
{\cal R}_{\cP} D &=& (k+1)S, \\
{\cal R}_{\cC} E\, &=& S+1, \\
{\cal R}_{\cP} Q &=& q^k,
\end{eqnarray*}
while  $x\, x^k = x^{k+1}$ and $x\, {x\choose k} = (k+1){x\choose k+1} + k {x\choose k}$ imply by (\ref{main}) and (\ref{RL}) that
\begin{eqnarray*}
{\cal R}_{\cP} X &=& S^{-1}, \\
{\cal R}_{\cC} X\,\! &=& k(1+S^{-1}).
\end{eqnarray*}
\end{example}

\begin{theorem}
\label{iso}
{ \cite[Prop.\,2 \& Thm.\,1]{APR}}
$\cL$ is a $\K$-algebra, and the transformation
$\cR: \cL \rightarrow \cE$,
defined in { (\ref{RL})}, is an isomorphism of $\K$-algebras.
\end{theorem}

Let $L \in \cL$ and $p\in \K[[\cB]]$ where $\cB = \langle P_k(x)\rangle_{k=0}^\infty$ is a factorial basis of $\K[x]$. Note that Theorem \ref{RBL} opens the way to finding solutions $y\in \K[[\cB]]$ of the equation $L y = p$ by the following three-step procedure:
\smallskip
\begin{center}
\large
\textbf{Procedure}\footnote{by a  \emph{procedure} we mean a \emph{high-level algorithm} where not all steps are fully specified yet} \textsc{DefiniteSumSols}
\end{center}
\label{defsumsols}

\begin{enumerate}
\item Compute $L' = \cR L \in \cE$.
\item Solve $L' c = \sigma_{\cB} p$ for the unknown $c \in \K^{\Z}$ with $c_k = 0$ for $k<0$.
\item Return $y(x) = \sum_{k=0}^\infty c_k P_k(x)$.
\end{enumerate}

\smallskip
As our goal is finding definite-sum solutions of linear recurrence equations, we henceforth limit our 
attention to linear \emph{\textbf{recurrence operators}} $L \in \K[x]\langle E\rangle$ and their associated 
operators $L' = \cR L\in \cE$ with respect to various factorial bases $\cB$ compatible with the shift 
operator $E$ (\emph{\textbf{shift-compatible bases}}, for short), and polynomial right-hand sides 
$p\in\K[x]$. In this case, we can use Definitions \ref{def:compat} and \ref{Lext}, Corollary \ref{cor_compat} and 
Proposition \ref{B=0_for_E} to elaborate step 1 of procedure \textsc{DefiniteSumSols} as follows:

\smallskip
\begin{center}
\label{alg:AssocOp}
\large
\textbf{Procedure} \textsc{AssociatedOp}
\label{AssOp}
\end{center}

\smallskip
\textsc{Input:} \verb+   + $L \in \K[x]\langle E\rangle$; 

\verb+          + a factorial basis $\cB = \langle P_k(x)\rangle_{k=0}^\infty$,  $(A,0)$-compatible with $E$

\medskip
\textsc{Output:} \verb+ + $L' = \cR L \in \cE$

\vskip 1pc
\begin{enumerate}
\item Using linear algebra, compute $\alpha_{k,-A},\alpha_{k,-A+1},\ldots, \alpha_{k,0} \in \K$ such that
   \begin{center}$P_k(x+1)\ = \ \sum_{i=-A}^0 \alpha_{k,i} P_{k+i}(x)$\  for all $k\in\N$.\end{center}
Let $E' = \sum_{i=0}^A \alpha_{k+i,-i} S^i\, \in\, \cE.$
\item Using linear algebra, compute $\beta_{k,0},\beta_{k,1} \in \K$ such that
   \begin{center}$x P_k(x)\ = \ \beta_{k,0}P_k(x) + \beta_{k,1}P_{k+1}(x)$.\end{center}
Let $X' = \beta_{k-1,1}S^{-1} + \beta_{k,0}\, \in\, \cE$.
\item Return the operator $L' \in \cE$, obtained from $L$ by substituting $E'$ for $E$ and $X'$ for $x$. 
\end{enumerate}

\section{The binomial transform and quasi-triangular bases}
\label{sec:BinTransf}

In the rest of the paper \emph{\textbf{we occasionally use $n$ instead of $x$}} to denote the independent variable of basis polynomials as well as of recurrence equations resp.\ operators. In particular, the shift operator $E$ acts both by $Ex = x+1$ and by $En = n+1$. 

In order for a formal-series solution $y(n) = \sum_{k=0}^\infty c_k P_k(n)$ of an equation $Ly=p$ obtained by procedure \textsc{DefiniteSumSols} given at the end of Section \ref{sec:FPS} to be a \textbf{\emph{definite-sum solution of our original equation}} where $L \in \K[x]\langle E\rangle$ and $y,p \in \K^\N$, we need to impose some additional requirements on the basis $\cB$. One obvious such requirement (satisfied, e.g., by the binomial-coefficient basis $\cC=\langle\binom{x}{k}\rangle_{k=0}^\infty$) is that it is \emph{locally finite}, meaning that for each $n\in\N$, there is an $f(n)\in\N$ such that $P_k(n)=0$ for all $k > f(n)$. If this is the case, we have $y(n) = \sum_{k=0}^{f(n)} c_k P_k(n) \in \K$, hence $y\in\K^\N$. Another desirable property of $\cC$ is its \emph{invertibility}, meaning that for each $a\in\K^\N$ there exists $b\in\K^\N$ such that $a_n = \sum_{k=0}^n b_k \binom{n}{k}$. In this section, we give some examples of computing $\cR L$ when $\cB = \cC$, and define the class of \emph{quasi-triangular bases} which are locally finite and invertible.


It follows from Example \ref{subst} and Theorem \ref{iso} that every linear recurrence operator $L \in \K[x]\langle E\rangle$ is compatible with the binomial-coefficient basis $\cC \ =\ \left\langle {x\choose k}\right\rangle_{k=0}^{\infty}$. To compute the associated operator $L' = {\cal R}_{\cC}L \in \cE$, we apply the substitution
\[
\begin{array}{lll}
E & \mapsto & S + 1, \\
x & \mapsto & k(1 + S^{-1})
\end{array}
\]
to all terms of $L \in \K[x]\langle E\rangle$. Clearly, $L' \in \K[k]\langle S, S^{-1}\rangle$,
and every $h \in \ker L'$ gives rise to a solution $y_n = \sum_{k=0}^\infty \binom{n}{k} c_k$ of $Ly = 0$. 

\begin{example}
\label{bigex}
Here we list some operators  $L \in \K[n]\langle E\rangle$, their associated operators $L' = {\cal R}_{\cC}L \in \cE$ with respect to the binomial-coefficient basis $\cC$, and some of the elements of their kernels. 
\begin{enumerate}
\item $L = E - c$ where $c \in \K^*$: Here $L' = S - (c - 1)$, and
\[
y_n\ =\ \sum_{k=0}^\infty \binom{n}{k} (c - 1)^k\ =\ c^n
\]
is indeed a solution of $Ly = 0$.

\item $L = E^2 - 2 E + 1$: Here $L' = S^2$, and by Theorem \ref{RBL}, any $c \in \ker L'$ satisfies $c_{n+2} = 0$ for all $n \ge 0$, or equivalently, $c_n = 0$ for all $n \ge 2$. Hence
\[
y_n\ =\ \sum_{k=0}^\infty \binom{n}{k} c_k\ =\ c_0 + c_1 n
\]
is indeed a solution of $Ly = 0$.

\item $L = E^2 - E - 1$: Here $L' = S^2 + S - 1$, and any $h \in\ker L'$ is of the form $c_n = (-1)^n (C_1 F_n + C_2 F_{n+1})$ where $C_1, C_2 \in \K$ and $F = \langle 0, 1, 1, 2, \ldots\rangle$ is the sequence of Fibonacci numbers. Hence every $y \in \ker L$ is of the form
\[
y_n\ =\ \sum_{k=0}^\infty \binom{n}{k}  (-1)^k (C_1 F_k + C_2 F_{k+1}).
\]
In particular, by setting $y_0 = F_0$ and $y_1 = F_1$, we find the identity
\[
F_n\ =\ \sum_{k=0}^n \binom{n}{k}  (-1)^{k+1} F_k.
\]

\item $L = E - (n+1)$: Here $L' = S - n - n S^{-1}$, and the equation to solve is
\[
c_{n+1} - n c_n - n c_{n-1}\ =\ 0\quad \text{ for } n\ge 0,
\]
which yields $c_1 = 0$ as well as
\[c_n = (n-1)(c_{n-1} + c_{n-2}) \quad \text{ for } n \ge 2.\]
The general solution of the latter equation is of the form 
\[
c_n\ =\ n! \left(C_1 \sum_{k=0}^n (-1)^k/k! + C_2\right)
\]
where $C_1, C_2 \in \K$ (cf.\ \cite[Example 8.6.1]{A=B}). Now $c_1 = 0$ implies $C_2 = 0$, hence every $y \in \ker L$ is of the form
\[
y_n\ =\ C \sum_{k=0}^\infty \binom{n}{k} k!\sum_{j=0}^k\frac{ (-1)^j}{j!}
\]
for some $C \in \K$. In particular, by setting $y_0 = 0! = 1 = C$, we find the identity
\[
n!\ =\ \sum_{k=0}^n \binom{n}{k} k! \sum_{j=0}^k\frac{ (-1)^j}{j!}
\]
or equivalently,
\[
\sum_{k=0}^n \frac{1}{k!} \sum_{j=0}^{n-k}\frac{ (-1)^j}{j!}\ =\ 1.
\]

\item $L =  E^3 - ( n^2 + 6 n + 10)E^2 + (n + 2) (2 n + 5) E - (n + 1) (n + 2)$: Unlike in the preceding four cases, the equation $L y = 0$ has no nonzero Liouvillian solutions 
(i.e., solutions which are interlacings of d'Alembertian sequences: cf.~\cite[Corollary~15.2]{Re12} or \cite[Theorem~12]{PZ13}). 
Here $L' =  S^3 - (n^2 + 6 n + 7)S^2 - (2 n^2 + 8 n + 7) S - (n + 1)^2$, and equation $L' h = 0$ has a hypergeometric solution $h_n = n!^2$. So 
\[
y_n\ =\  \sum_{k=0}^\infty \binom{n}{k} k!^2\ =\  \sum_{k=0}^n \binom{n}{k} k!^2
\]
is a non-Liouvillian definite-sum solution of equation $L y = 0$. 

\end{enumerate}
\end{example}

\begin{remark}
Since $\cRc(n^{\underline{i}})=k^{\underline{i}} \sum_{j=0}^i \binom{i}{j} S^{-j}$ (as can be easily seen by induction on $i$), every negative power $S^{-j}$ in $L'$ is multiplied by $k^{\underline{i}}$ for some $i \ge j$. So all terms of $L'$ containing $ S^{-j}$ vanish for $k = 0, 1, \dots, j-1$ (cf.~the term $k S^{-1}$, renamed as $n S^{-1}$ in $L'$ of item 4 in Example \ref{bigex}).
\end{remark}

Note that \emph{any} sequence $a \in \K^\N$ can be represented in the form
\begin{equation}
\label{bintransf}
a_n\ =\ \sum_{k=0}^\infty \binom{n}{k} b_k
\quad \text{ for all } n \in \N
\end{equation}
where $b \in \K^\N$ satisfies
\[
b_n\ =\ \sum_{k=0}^\infty (-1)^{n-k}\binom{n}{k} a_k
\quad \text{ for all } n \in \N. 
\]
This is because the infinite matrix $B=\left[\binom{n}{k}\right]_{n,k=0}^\infty$ of the system of linear equations (\ref{bintransf}) for the unknown $b_0, b_1, b_2, \dots$ is lower triangular with unit diagonal, so it is invertible, and it is easy to see that its inverse is
$B^{-1}\ =\ \left[(-1)^{n-k}\binom{n}{k}\right]_{n,k=0}^\infty$.
Some authors call the sequence $a$ the \emph{binomial transform} of $b$, and $b$ the \emph{inverse binomial transform} of $a$ (cf.\ \cite[seq.~A007317]{oeis}). Others define the binomial transform as an involution: $a_n\ =\ \sum_{k=0}^\infty (-1)^k \binom{n}{k} b_k$, and $b_n\ =\ \sum_{k=0}^\infty (-1)^k \binom{n}{k} a_k$ (cf.\ \cite[p.~137, Exercise 36]{KnuthACP3}). For our purposes, the actual values of leading coefficients of the basis elements are not important.

\begin{proposition}
Let $L = \sum_{i=0}^r p_i(n) E^i \in \K[n]\langle E\rangle$, and let its associated operator w.r.t.\ the binomial-coefficient basis $\cC$ be $L' = \cRc(L) = \sum_{i=-t}^s q_i(k) S^i \in \K[k]\langle S, S^{-1}\rangle$. Then 
\begin{enumerate}
\item $s = r$, and
\item $q_s(k)S^s = p_r(k)S^r$.
\end{enumerate}
In other words, $\ord L' = \ord L$, and the leading term of $L'$ (after renaming $k\to n$, $S\to E$) agrees with that of $L$.
\end{proposition}

\begin{proof}
Clearly
\begin{align*}
\cRc(n^j) &= (k + k S^{-1})^j = k^j +\mathcal{O}(S^{-1}), \\
\cRc(E^i) &= (S+1)^i = S^i + \mathcal{O}(S^{i-1}),
\end{align*}
where $\mathcal{O}(S^\rho)$ denotes an operator from $\K[k]\langle E, E^{-1}\rangle$ of order at most $\rho$. Hence for $p_i(n) = \sum_{j=0}^{d_i} c_{i,j}n^j$ we have
\begin{align*}
\cRc(p_i(n))\ &= \sum_{j=0}^{d_i} c_{i,j}\cRc(n^j) = \sum_{j=0}^{d_i} c_{i,j}( k^j + \mathcal{O}(S^{-1})) = p_i(k) + \mathcal{O}(S^{-1}),\\
\cRc(p_i(n) E^i)\ &= \left(p_i(k) + \mathcal{O}(S^{-1})\right)\left(S^i + \mathcal{O}(S^{i-1})\right) = p_i(k)S^i + \mathcal{O}(S^{i-1})
\end{align*}
and
\begin{align*}
\cRc(L)\ &= \cRc\left(\sum_{i=0}^r p_i(n) E^i\right) = \sum_{i=0}^r \cRc\left(p_i(n)E^i\right) = \sum_{i=0}^r \left(p_i(k)S^i + \mathcal{O}(S^{i-1})\right)\\
&= p_r(k)S^r +  \mathcal{O}(S^{r-1}) + \sum_{i=0}^{r-1} \mathcal{O}(S^i) = p_r(k)S^r +  \mathcal{O}(S^{r-1}),
\end{align*}
proving the claim. 
\end{proof}

\smallskip
Beside the binomial-coefficient basis $\cC = \langle \binom{n}{k} \rangle_{k=0}^{\infty}$, there are many other shift-compatible bases $\cB = \langle P_k(n) \rangle_{k=0}^{\infty}$ with the property that any sequence $a \in \K^\N$ can be represented in the form
\begin{align}
\label{qttransf}
a_n\ =\ \sum_{k=0}^\infty b_k P_k(n)
\end{align}
 for some $b \in \K^\N$.

\begin{definition}
\label{qt}
Call a shift-compatible basis $\cB = \langle P_k(n) \rangle_{k=0}^{\infty}$ \emph{quasi-triangular} if there is a strictly increasing function $f\!: \N \to \N$ such that
\begin{enumerate}
\item $\forall k,n \in \N\!: (k > f(n)\ \Longrightarrow\ P_k(n) = 0)$,
\item $\forall n \in \N\!: P_{f(n)}(n) \ne 0$. 
\end{enumerate}
\end{definition}
Clearly, the basis $\cC = \langle \binom{n}{k} \rangle_{k=0}^{\infty}$ is quasi-triangular with $f(n)=n$.

\begin{proposition}\label{prop:quasi}
A basis $\cB = \langle P_k(n) \rangle_{k=0}^{\infty}$ is quasi-triangular if and only if its root sequence $\rho = \langle \rho_1,\rho_2,\rho_3,\dots\rangle$ satisfies
\begin{enumerate}
\item $\langle 0,1,2,3,\dots\rangle$ is a subsequence of $\rho$,
\item for every $n\in\N$, the first appearance of $n$ in $\rho$ precedes the first appearance of $n+1$ in $\rho$.
\end{enumerate}
\end{proposition}

\begin{proof}
Assume first that $\cB$ is quasi-triangular with $f\!: \N \to \N$ as in Definition~\ref{qt}, and let $n\in\N$. Then $P_{f(n)}(n) \ne 0$ and $P_{f(n)+1}(n) = 0$, hence $\rho_1,\rho_2,\dots,\rho_{f(n)} \ne n$ and $\rho_{f(n)+1} = n$. Since $f$ is strictly increasing, $\langle \rho_{f(0)+1}, \rho_{f(1)+1}, \rho_{f(2)+1}, \dots\rangle = \langle 0,1,2,\dots\rangle$ is a subsequence of $\rho$, proving item 1. Since $P_{f(n)}(n) \ne 0$ and $\cB$ is factorial, $P_k(n)\ne 0$ for all $k\le f(n)$, so $\rho_1,\rho_2,\dots,\rho_{f(n)} \ne n$, hence the first term of $\rho$ equal to $n$ is $\rho_{f(n)+1}$. As $f$ is strictly increasing, this proves item 2.

Assume now that the root sequence $\rho$ of $\cB$ satisfies items 1 and 2, and let $f(n) := \min \{k\in\N\setminus\{0\};\ \rho_k=n\} - 1$. Then $f\!: \N \to \N$ is strictly increasing, $\rho_1,\rho_2,\dots,\rho_{f(n)} \ne n$, while $\rho_{f(n)+1} = n$, so $P_k(n) = 0$ for all $k>f(n)$, and $P_{f(n)}(n) \ne 0$, hence $\cB$ is quasi-triangular. 
\end{proof}

\begin{theorem}
Let $\cB = \langle P_k(n) \rangle_{k=0}^{\infty}$ be a quasi-triangular basis with $f\!: \N \to \N$ as in Definition~\ref{qt}. Then for every $a \in \K^\N$ there exists $b \in \K^\N$ such that \emph{(\ref{qttransf})} holds.
\end{theorem}

\begin{proof}
Since $f$ is strictly increasing, it is injective, and we can define $b_k\in\K$ for $k=0,1,2,\dots$ recursively as follows:
\begin{enumerate}
\item If $k=f(n)$ for some $n\in\N$ then let $b_k\ =\ \frac{1}{P_k(n)}\left(a_n - \sum_{i=0}^{k-1}b_i P_i(n)\right)$.
\item If $k \notin f(\N)$ then let $b_k \in\K$ be arbitrary.
\end{enumerate}
Then for every $n\in\N$ we have $b_{f(n)}\ =\ \frac{1}{P_{f(n)}(n)}\left(a_n - \sum_{i=0}^{f(n)-1}b_i P_i(n)\right)$, hence
\begin{align*}
a_n\ =\ b_{f(n)}P_{f(n)}(n)\ + \!\sum_{i=0}^{f(n)-1}b_i P_i(n)\ =\ \sum_{k=0}^{f(n)}b_k P_k(n)\ =\ \sum_{k=0}^{\infty}b_k P_k(n),
\end{align*}
proving equality (\ref{qttransf}). 
\end{proof}

\section{Products of compatible bases}
\label{sec:ProdBases}

To be able to use formal polynomial series to find other definite-sum solutions of linear recurrence equations, we need a rich supply of shift-compatible bases. 

\begin{definition}
\label{def:cab} 
For $a \in \N \setminus \{0\}$, $b \in \K$, and for all $k \in \N$, let $P_k^{(a,b)}(n) := \binom{an+b}{k}$. We denote the polynomial basis $\left\langle P_k^{(a,b)}(n) \right\rangle_{k=0}^{\infty}$ by  $\cC_{a,b}$, and call it a \emph{generalized binomial-coefficient basis} of $\K[n]$.
\end{definition}

\begin{proposition}
\label{cab}
Any generalized binomial-coefficient basis $\cC_{a,b}$ is a factorial basis of $\K[n]$, which is $(a,0)$-compatible with the shift operator $E$. If $b\in\N$ then $\cC_{a,b}$ is quasi-triangular.
\end{proposition}

\begin{proof}
Clearly $\deg_n P_k^{(a,b)}(n) = k$ and 
\[
P_{k+1}^{(a,b)}(n) = \frac{an+b-k}{k+1} P_k^{(a,b)}(n)\ \ { for\ all\ } k \in \N,
\]
so $P_k^{(a,b)}(n) \,\big|\, P_{k+1}^{(a,b)}(n)$, and $\cC_{a,b}$ is factorial. By Chu-Vandermonde's identity, 
\begin{eqnarray*}
E P_k^{(a,b)}(n) &=&  P_k^{(a,b)}(n+1)\ =\ \binom{an+a+b}{k}\ =\ \,\sum_{i=0}^a \binom{a}{i}\binom{an+b}{k-i} \\
&=& \sum_{i=-a}^0 \binom{a}{-i}\binom{an+b}{k+i}\ =\ \sum_{i=-a}^0 \binom{a}{-i}P_{k+i}^{(a,b)}(n),
\end{eqnarray*}
so $\cC_{a,b}$ is $(a,0)$-compatible with $E$
($\alpha_{k,i} = \binom{a}{-i}$ for $i = -a, -a+1, \ldots, 0$). Finally, if $b\in\N$, let $f(n) = a n + b$. Then $f:\N\to\N$ is strictly increasing, $P_k^{(a,b)}(n) = \binom{an+b}{k} = 0$ for $k > an+b = f(n)$, and $P_{f(n)}^{(a,b)}(n) = \binom{an+b}{an+b} = 1$, so $\cC_{a,b}$ is quasi-triangular by Definition \ref{qt}.

\end{proof}

\begin{definition}\label{def:pbasis}

Let $m \in \N \setminus \{0\}$, and for $i = 1, 2, \ldots, m$, let $\cB_i = \langle P_k^{(i)}(n) \rangle_{k=0}^\infty$ be a basis of $\K[n]$. For all $k \in \N$ and $j \in \{0,1,\ldots,m-1\}$, let
\begin{equation}\label{equ:product}
P_{m k + j}^{(\pi)}(n)\ :=\ \prod_{i=1}^j P_{k+1}^{(i)}(n)\cdot \prod_{i=j+1}^m  P_k^{(i)}(n).
\end{equation}
Then the sequence $\prod_{i=1}^m \cB_i := \langle P_n^{(\pi)}(n) \rangle_{n=0}^\infty$ is the \emph{product} of $\cB_1, \cB_2, \ldots, \cB_m$. 
\end{definition}

\begin{theorem}
\label{prod}
Let $\cB_1, \cB_2, \ldots, \cB_m$ be factorial bases of $\K[n]$, and $L \in {\cal L}_{\K[n]}$.
\begin{enumerate}
\item $\prod_{i=1}^m \cB_i$ is a factorial basis of $\K[n]$.
\item Let $L$ be a ring endomorphism of $\K[n]$, and let all $\cB_i$ be $(A_i, B_i)$-compatible with $L$.  
   Write $A = \max_{1 \le i \le m} A_i$ and $B = \min_{1 \le i \le m} B_i$. Then $\prod_{i=1}^m \cB_i$ is 
   $(m A, B)$-compatible with $L$.
\end{enumerate}
\end{theorem}

\begin{proof}
\begin{enumerate}
\item Clearly $\deg P_{m k + j}^{(\pi)}(n) =  j(k+1) + (m-j)k = mk+j$.

If $\ell = mk+j$ with $0 \le j \le m-2$, then $\ell+1 = mk+(j+1)$ and
\[
\frac{P_{\ell+1}^{(\pi)}(n)}{P_\ell^{(\pi)}(n)}\ =\ \frac{\prod_{i=1}^{j+1} P_{k+1}^{(i)}(n)}{\prod_{i=1}^j P_{k+1}^{(i)}(n)} \cdot \frac{\prod_{i=j+2}^m  P_k^{(i)}(n)}{\prod_{i=j+1}^m  P_k^{(i)}(n)}\ =\ \frac{ P_{k+1}^{(j+1)}(n)}{P_k^{(j+1)}(n)} 
\ \in\ \K[n]
\]
as $\cB_{j+1}$ is factorial. If $\ell = mk+(m-1)$, then $\ell+1 = m(k+1) + 0$ and
\[
\frac{P_{\ell+1}^{(\pi)}(n)}{P_\ell^{(\pi)}(n)}\ =\ \frac{1}{\prod_{i=1}^{m-1} P_{k+1}^{(i)}(n)} \cdot \frac{\prod_{i=1}^m  P_{k+1}^{(i)}(n)}{P_k^{(m)}(n)}\ =\ \frac{ P_{k+1}^{(m)}(n)}{P_k^{(m)}(n)} 
\ \in\ \K[n]
\]
because $\cB_{m}$ is factorial. Hence $\prod_{i=1}^m \cB_i$ is factorial as well.

\item Let $p \in \K[n]$ be arbitrary. For $i = 1,2,\ldots,m$, let $p = \sum_{k=0}^{\deg p} c_k^{(i)} P_k^{(i)}$ be the expansion of $p$ w.r.t.\ $\cB_i$. Then  $L p = \sum_{k=0}^{\deg p} c_k^{(i)} L P_k^{(i)}$, and by condition \textbf{C1} of Proposition \ref{compat},
\[
\deg L p \le \max_{0 \le k \le \deg p} \deg L P_k^{(i)} \le \max_{0 \le k \le \deg p} (k + B_i) = \deg p + B_i.
\]
Since this holds for all $i$, we have $\deg L p\ \le\ \deg p + B$ for all $p \in \K[n]$. In particular, $\deg L P_k^{(\pi)}\ \le\ k + B$, so $\prod_{i=1}^m \cB_i$ satisfies \textbf{C1}.\\[2pt]

Condition \textbf{C2} of Proposition \ref{compat} and our definition of $A$ imply that\\ $P_{k+1-A}^{(i)} \,\big|\, L P_{k+1}^{(i)}$ and $P_{k-A}^{(i)} \,\big|\, L P_{k}^{(i)}$ for all $k \ge A$ and $i \in \{1,2\ldots,m\}$,~so
\[
 P_{m (k-A) + j}^{(\pi)}\ =\ \prod_{i=1}^j P_{k+1-A}^{(i)}\cdot \prod_{i=j+1}^m  P_{k-A}^{(i)} \ \ \bigg|\ \ \prod_{i=1}^j L P_{k+1}^{(i)}\cdot \prod_{i=j+1}^m  L P_k^{(i)},
\]
or equivalently, since $L$ is an endomorphism of the ring $\K[n]$,
\[
 P_{m (k-A) + j}^{(\pi)} \ \bigg|\  L \left(\prod_{i=1}^j P_{k+1}^{(i)}\cdot \prod_{i=j+1}^m  P_k^{(i)}\right) \ =\ L  P_{m k + j}^{(\pi)}. 
\]
For $\ell = mk+j \ge m A$, this turns into $P_{\ell - m A}^{(\pi)}\ \big|\,\ L  P_{\ell}^{(\pi)}$, so $\prod_{i=1}^m \cB_i$ satisfies \textbf{C2} as well. By Proposition \ref{compat}, this proves the claim.  

\end{enumerate}
\end{proof}

\begin{definition}
\label{def:caabb} 
Let $m \in \N \setminus \{0\}$, and let $\mathbf{a} = (a_1,a_2,\ldots,a_m)$, $\mathbf{b} = (b_1,b_2,\ldots,b_m)$ where $a_i \in \N \setminus \{0\}$, $b_i \in \K$ for $i = 1,2,\ldots,m$. We denote the product of generalized binomial-coefficient bases $\prod_{i=1}^m \cC_{a_i, b_i}$ by $\cC_{\mathbf{a},\mathbf{b}}$, and call it a \emph{product binomial-coefficient basis} of $\K[x]$ having \emph{length} $m$. 
\end{definition}

\begin{corollary}
\label{cor:caabb} 
Any product binomial-coefficient basis $\cC_{\mathbf{a},\mathbf{b}}$  is a factorial basis of $\K[x]$ which is $(m a, 0)$-compatible with $E$, where $a = \max_{1\le i\le m}  a_i$.
\end{corollary}

\begin{proof}
\noindent
Use Theorem \ref{prod} and Proposition \ref{cab}. 
\end{proof}

\medskip
By Corollary \ref{cor:caabb}, we now have at our disposal a family of factorial bases  $\cC_{\mathbf{a},\mathbf{b}}$ which, given an operator $L \in \K[n]\langle E\rangle$ and a kernel $K(n,k)$ of the form
\begin{align}
\label{binom}
K(n,k) \ =\ \prod_{i=1}^m \binom{a_i n + b_i}{k}
\end{align}
where $a_i \in \N \setminus \{0\}$ and $b_i \in \K$, 
can be used to find solutions of $L y = 0$ having the form $y_n = \sum_{k=0}^\infty K(n,k)\, c_k$. To this end, we need to compute expansions of $E P_n(x)$ and $X P_n(x)$ in the basis  $\cC_{\mathbf{a},\mathbf{b}}=  \langle P_n(x)\rangle_{n=0}^\infty$.

\begin{example}
\label{xk2E}

Take $K(n,k) = \binom{n}{k}^2$. The polynomial basis to be used here is $\cC_{(1,1),(0,0)} = \langle P_n(x)\rangle_{n=0}^\infty$ where for all $k \in \N$,
\[
P_{2k}(x) =  \binom{x}{k}^2, \quad P_{2k+1}(x) =  \binom{x}{k+1}  \binom{x}{k}.
\]
According to Corollary \ref{cor:caabb}, $\cC_{(1,1),(0,0)}$ is a factorial basis of $\K[x]$ with $m = 2$ and $a = \max\{1,1\} = 1$, so it is $(2,0)$-compatible with $E$. In particular, this means that $P_{2k}(x+1)$ can be expressed as a linear combination of $P_{2k}(x)$, $P_{2k-1}(x)$, $P_{2k-2}(x)$, and $P_{2k+1}(x+1)$ as a linear combination of $P_{2k+1}(x)$, $P_{2k}(x)$, $P_{2k-1}(x)$, with coefficients depending on $k$. In the case of $P_{2k}(x+1)$ this is just an application of Pascal's rule:
\begin{eqnarray}
P_{2k}(x+1) &=&  \binom{x+1}{k}^2\ =\ \left[\binom{x}{k} +  \binom{x}{k-1}\right]^2 \nonumber\\
&=&  \binom{x}{k}^2 +  2 \binom{x}{k}\binom{x}{k-1} + \binom{x}{k-1}^2 \nonumber\\[3pt]
&=& P_{2k}(x) + 2  P_{2k-1}(x) +  P_{2k-2}(x). \label{p2k}
\end{eqnarray}
 In the case of $P_{2k+1}(x+1)$, we can use the method of undetermined coefficients. Dividing both sides of
\begin{eqnarray*}
P_{2k+1}(x+1) &=&  u(k) P_{2k+1}(x) + v(k)  P_{2k}(x) +  w(k) P_{2k-1}(x), {\ \ or} \\
\binom{x+1}{k+1} \binom{x+1}{k} &=& u(k) \binom{x}{k+1}\binom{x}{k} +  v(k) \binom{x}{k}^2 + w(k) \binom{x}{k}\binom{x}{k-1}
\end{eqnarray*}
where $u(k), v(k), w(k)$ are undetermined functions of $k$, by $\binom{x}{k}\binom{x}{k-1}$, yields
\begin{align}
\label{undetcoef}
\frac{(x+1)^2}{k(k+1)}\ =\ u(k) \frac{(x-k+1)(x-k)}{k(k+1)} + v(k) \frac{x-k+1}{k} + w(k),
\end{align}
which is an equality of two quadratic polynomials from $\K(k)[x]$. Plugging in the values $x = -1, k, k-1$, we obtain a triangular system of linear equations
\[
\begin{array}{rcc}
u(k) - v(k) + w(k) &=& 0 \\[2pt]
\frac{1}{k} v(k) + w(k) &=& \frac{k+1}{k} \\[3pt]
w(k) &=& \frac{k}{k+1}
\end{array}
\]
whose solution is $u(k) = 1$ (as expected), $v(k) = \frac{2k+1}{k+1}$,  $w(k) = \frac{k}{k+1}$, and so
\begin{equation}
\label{p2k1}
P_{2k+1}(x+1)\ =\  P_{2k+1}(x) + \frac{2k+1}{k+1}  P_{2k}(x) +  \frac{k}{k+1} P_{2k-1}(x).
\end{equation}
Alternatively, we could obtain a system of linear equations for $u(k), v(k), w(k)$ by equating the coefficients of $x^j$ on both sides of (\ref{undetcoef}) for $j=0,1,2$.

For the expansion of $X P_n(x)$, recall that by Corollary \ref{cor_compat} every factorial basis is $(0,1)$-compatible with $X$. Indeed, as $x \binom{x}{k} = (k+1)\binom{x}{k+1} +  k\binom{x}{k}$, we have
\begin{align*}
x P_{2k}(x)\ =\ x \binom{x}{k}^2\ &=\ (k+1)\binom{x}{k+1}\binom{x}{k} +  k\binom{x}{k}^2\\
&=\ (k+1) P_{2k+1}(x) + k P_{2k}(x), \\
x P_{2k+1}(x)\ =\ x\binom{x}{k+1}\binom{x}{k}\ &=\ (k+1)\binom{x}{k+1}^2 +  k\binom{x}{k+1}\binom{x}{k}\\
&=\ (k+1) P_{2k+2}(x) + k P_{2k+1}(x).
\end{align*}
It is easy to see that $\cC_{(1,1),(0,0)}$ is quasi-triangular with $f(n) = 2n$.

\end{example}

For additional examples of expansions of shifted basis elements in the basis $\cC_{\mathbf{a},\mathbf{b}}$, see \cite{arXiv}.

\medskip
\label{expansion}
If our kernel is as in (\ref{binom}), we can use the product binomial-coefficient basis $\cC_{\mathbf{a},\mathbf{b}} = \langle P_n^{(\pi)}(x)\rangle_{n=0}^\infty$ which, by Corollary \ref{cor:caabb}, is $(m a, 0)$-compatible with $E$ where $a = \max_{1\le i\le m} a_i$. In order to compute $\alpha_{k,j,i} \in \K(k)$ such that
\begin{align}
\label{koefs}
P_{mk+j}^{(\pi)}(x+1)\ =\ \sum_{i=-ma}^0 \alpha_{k,j,i} P_{mk+j+i}^{(\pi)}(x)
\end{align}
for all $k \in \N$ and $j \in \{0,1,\ldots,m-1\}$, we divide both sides of this equation by $P_{mk+j-ma}^{(\pi)}(x)$ which turns it into an equality of two polynomials of degree $ma$ from $\K(k)[x]$. From this equality a system of $ma+1$ linear algebraic equations for the $ma+1$ undetermined coefficients $\alpha_{k,j,i}$, $i = 0, 1, \ldots, ma$, can be obtained by equating the coefficients of like powers of $x$ on both sides, or (as in Example \ref{xk2E}) by substituting $ma+1$ distinct values from $\K(k)$ for $x$ in this equality. Note that for each $j \in \{0,1,\ldots,m-1\}$, this system is uniquely solvable since $\cC_{\mathbf{a},\mathbf{b}}$ is a basis of $\K[x]$, that the $\alpha_{k,j,i}$ will be rational functions of $k$, and that, as the shift operator preserves leading coefficients and degrees of polynomials, $\alpha_{k,j,0} = 1$.

To compute the coefficients of the expansion of $x P_n^{(\pi)}(x)$ w.r.t.\ $\cC_{\mathbf{a},\mathbf{b}}$, we use the fact that by Corollary \ref{cor_compat}, $\cC_{\mathbf{a},\mathbf{b}}$ is $(0,1)$-compatible with $X$:

\begin{proposition}
\label{x}
For $k \in \N$ and $j \in \{0,1,\ldots,m-1\}$, let
\begin{equation}
\label{Pmkj}
P_{m k + j}^{(\pi)}(x)\ :=\ \prod_{i=1}^j \binom{a_i x+b_i}{k+1}\cdot \prod_{i=j+1}^m \binom{a_i x+b_i}{k}.
\end{equation}
Then
\[
x P_{m k + j}^{(\pi)}(x)\ =\ \frac{k+1}{a_{j+1}} P_{m k + j + 1}^{(\pi)}(x) + \frac{k-b_{j+1}}{a_{j+1}} P_{m k + j}^{(\pi)}(x).
\]
\end{proposition}

\begin{proof}
\[
\frac{k+1}{a_{j+1}} P_{m k + j + 1}^{(\pi)}(x) + \frac{k-b_{j+1}}{a_{j+1}} P_{m k + j}^{(\pi)}(x)\ =\ P_{m k + j}^{(\pi)}(x) \cdot f(x)
\]
where
\begin{eqnarray*}
f(x) &=& \frac{k+1}{a_{j+1}} \cdot \frac{P_{m k + j + 1}^{(\pi)}(x)}{P_{m k + j}^{(\pi)}(x)} + \frac{k-b_{j+1}}{a_{j+1}}\\
&=& \frac{k+1}{a_{j+1}} \cdot \frac{\binom{a_{j+1} x+b_{j+1}}{k+1}}{\binom{a_{j+1} x+b_{j+1}}{k}} + \frac{k-b_{j+1}}{a_{j+1}}\\ 
&=& \frac{k+1}{a_{j+1}} \cdot \frac{a_{j+1} x + b_{j+1} - k}{k+1} + \frac{k-b_{j+1}}{a_{j+1}}\ =\ x.
\end{eqnarray*}
\end{proof}

Examples of factorial shift-compatible quasi-triangular polynomial bases that are \emph{not} of the type $\cC_{\mathbf{a},\mathbf{b}}$ are given in Section \ref{sec:ex}.

\section{Sieved polynomial bases}
\label{sec:Sieved}

Now we can use 
Procedure \textsc{AssociatedOp} on p.~\pageref{alg:AssocOp} to find the associated operator $\cR L$ where $\cB = \cC_{\mathbf{a},\mathbf{b}}$. 
Notice however that for $m > 1$, the coefficients $\alpha_{k,i}$ expressing the actions of $E$ resp.\ $X$ on $\cB$ are not rational functions of $k$ anymore, but conditional expressions evaluating to $m$ generally distinct rational functions, depending on the residue class of $k \bmod m$ (cf.\ Example \ref{xk2E} with $m=2$, and Proposition \ref{x}). So the coefficients of $\cR L$, obtained by composing and adding the operators $\cR E$ and $\cR X$ repeatedly, will contain quite complicated conditional expressions. In addition, $\ord \cR L$ may exceed $\ord L$ by a factor of $ma$ which can be exponential in input size. 

To overcome these inconveniences, we note that product bases represent a special case of \textbf{\emph{sieved polynomial bases}} where the definition of the $k$-th basis element $P_k(x)$ depends on the residue of $k$ modulo some $m \in \N$, $m \ge 1$ (for similar phenomena in the theory of orthogonal polynomials satisfying three-term recurrences, cf.\ \cite{sieved} and the series of papers  \cite{Ism8}--\cite{Ism10}, \cite{Ism1}--\cite{Ism9}). For a sieved basis $\cB$ with modulus $m$ (an \emph{$m$-sieved basis}, for short) we do not attempt to compute $\cR L$ directly but represent it by a \textbf{\emph{matrix $\left[\cR L\right] = \left[L_{r,j}\right]_{r,j=0}^{m-1}$  of operators}} where \emph{$L_{r,j} \in \cE$ expresses the contribution of the $j$-th $m$-section  $s_j^m \sigma_\cB y$ of the coefficient sequence of $y$ to the $r$-th $m$-section $s_r^m \sigma_\cB (L y)$ of the coefficient sequence of $Ly$} (see Definition~\ref{def1} and Notation \ref{notL_B} for the definitions of $s_j^m$ and $\sigma_\cB$, $\cE$, respectively). Note that being $m$-sieved is not an intrinsic property of a polynomial basis, but rather describes its presentation.

\begin{proposition}
\label{mainprop}
Let $L \in {\cal L}_{\K[x]}$,  $\cB = \langle P_n(x)\rangle_{n=0}^\infty$ (a factorial basis of $\K[x]$), $m \in \N \setminus \{0\}$, $A, B \in \N$ and $\alpha_{k,j,i} \in \K$ be such that for all $k \in \N$ and $j \in \{0,1,\ldots,m-1\}$,
\begin{equation}
\label{LPskj}
L P_{mk+j}(x) = \sum_{i=-A}^B \alpha_{k,j,i} P_{mk+j+i}(x).
\end{equation}
Furthermore, for all $r,j \in \{0,1,\ldots,m-1\}$ define
\begin{eqnarray}
\label{Lrjdef}
\ L_{r,j}\ := \!\!\!\!\!\sum_{\,\myatop{A \le i \le B}{i+j \equiv r \pmod{\!m}}} \!\!\!\!\!\!\!\!\!\!\alpha_{k+\frac{r-i-j}{m},j,i} S^{\frac{r-i-j}{m}} \in\ \cE 
\end{eqnarray}
(to keep notation simple, we do not make the dependence of $ L_{r,j}$ on $m$ explicit).
Then for every $y \in \K[[\cB]]$  and $r \in \{0,1,\ldots,m-1\}$,
\begin{equation}
\label{Lrj}
 s_r^m \sigma_\cB (L y)\ =\ \sum_{j=0}^{m-1} L_{r,j}\,  s_j^m \sigma_\cB y.
\end{equation}
\end{proposition}

\begin{proof}
Write $y(x) = \sum_{n=0}^\infty c_n P_n(x)$ as the sum of its $m$-sections
\begin{equation}
\label{y}
y(x) = \sum_{j=0}^{m-1} \sum_{k=0}^\infty c_{mk+j} P_{mk+j}(x) = \sum_{j=0}^{m-1} \sum_{k=0}^\infty (s_j^mc)_k P_{mk+j}(x).
\end{equation}
Then
\begin{align}
L y(x) &=\  \sum_{j=0}^{m-1} \sum_{k=0}^\infty (s_j^mc)_k L P_{mk+j}(x) = \sum_{j=0}^{m-1} \sum_{k=0}^\infty  (s_j^mc)_k \sum_{i=-A}^B \alpha_{k,j,i} P_{mk+j+i}(x) \label{Ly(x)} \\
&=\ \sum_{j=0}^{m-1}\sum_{r=0}^{m-1} \sum_{\myatop{\,-A \le i \le B}{\, i+j \equiv r \pmod{\!m}} } \sum_{k=0}^\infty \alpha_{k,j,i}\, \left(s_j^mc\right)_k\,  P_{mk+i+j}(x) \label{residue} \\
&= \sum_{r,j=0}^{m-1} \sum_{\myatop{\,-A \le i \le B }{\, i+j \equiv r \pmod{\!m}}} \sum_{k=\frac{i+j-r}{m}}^\infty \alpha_{k+\frac{r-i-j}{m},j,i}\, \left(s_j^mc\right)_{k+\frac{r-i-j}{m}}\,  P_{mk+r}(x)\label{subsk} \\
&= \sum_{r,j=0}^{m-1} \sum_{\myatop{\,-A \le i \le B }{\, i+j \equiv r \pmod{\!m}}} \sum_{k=0}^\infty \alpha_{k+\frac{r-i-j}{m},j,i}\left(S^{\frac{r-i-j}{m}}s_j^mc\right)_k  P_{mk+r}(x) \label{zero}\\
&=\ \sum_{r=0}^{m-1} \sum_{k=0}^\infty \bigg(\sum_{j=0}^{m-1} L_{r,j}\, s_j^mc\bigg)_k\,  P_{mk+r}(x) \label{final}
\end{align}
where in (\ref{Ly(x)}) we used  (\ref{LPskj}), in (\ref{residue}) we reordered summation on $i$ with respect to the residue class of $i+j \bmod m$, (\ref{subsk}) was obtained by replacing $k$ with $k+\frac{r-i-j}{m}$, (\ref{zero}) by noting that 
\begin{align*}
k < 0\ \ &\Longrightarrow\ \ mk+r < 0\ \ \Longrightarrow\ \ P_{mk+r}=0, \\
k < \frac{i+j-r}{m}\ \ &\Longrightarrow\ \ \left(s_j^mc\right)_{k+\frac{r-i-j}{m}} = \left(s_j^mc\right)_{k - \frac{i+j-r}{m}} = 0,
\end{align*}
and (\ref{final}) by using (\ref{Lrjdef}). Now the equality of $Ly(x)$ and the series in (\ref{final}) can be restated as (\ref{Lrj}). 
\end{proof}

\begin{corollary}
\label{cor}
Let $L$, $\cB$, $m$, and $L_{r,j}$ for $r \in \{0,1,\ldots,m-1\}$ be as in Proposition \ref{mainprop}, and let $p \in \KK$. Then 
\[
L y = p\quad \Longleftrightarrow\quad  \forall r \in \{0,1,\ldots,m-1\}\!: \sum_{j=0}^{m-1} L_{r,j}\,  s_j^m \sigma_\cB y\ =\ s_r^m \sigma_\cB p.
\]
\end{corollary}

\begin{proof}
By Proposition \ref{mainprop},
\begin{align*}
L y = p\ &\Longleftrightarrow\ \sigma_\cB(L y) = \sigma_\cB p \\
&\Longleftrightarrow\ \forall r \in \{0,1,\ldots,m-1\}\!: s_r^m \sigma_\cB(L y) = s_r^m \sigma_\cB p  \\
&\Longleftrightarrow\ \forall r \in \{0,1,\ldots,m-1\}\!: \sum_{j=0}^{m-1} L_{r,j} s_j^m \sigma_\cB y = s_r^m \sigma_\cB p. \ \quad\qquad\qquad\quad \Box  
\end{align*}
\end{proof}

\smallskip
Note that for $m=1$, Proposition \ref{mainprop} and Corollary \ref{cor} turn into Theorem \ref{RBL}.1 and Theorem \ref{RBL}.2, respectively (with $\alpha_{k,i} = \alpha_{k,0,i}$ and $\cR L = L_{0,0}$).

\begin{notation} $\left[\cR L\right] := \left[L_{r,j}\right]_{r,j=0}^{m-1} \in M_m(\cE)$ where $L_{r,j}$ is as given in (\ref{Lrjdef}).
\end{notation}

\begin{proposition}
\label{matrixhom}
Let $L^{(1)}, L^{(2)} \in \cL$. Then
\[
\left[\cR \left(L^{(1)} L^{(2)}\right)\right]\ =\ \left[\cR L^{(1)}\right]\!\!\left[\cR L^{(2)}\right].
\]
\end{proposition}

\begin{proof}
Write $L = L^{(1)} L^{(2)}$ and $\sigma = \sigma_\cB$. By (\ref{Lrj}),
\begin{eqnarray*}
s_t^m \sigma (L y) \ =\ \sum_{j=0}^{m-1} L_{t,j}\,  s_j^m \sigma y\ =\  \sum_{j=0}^{m-1}\left[\cR L\right]_{t,j} s_j^m \sigma y.
\end{eqnarray*}
On the other hand, by (\ref{Lrj}) applied to $L^{(1)}$ and $L^{(2)}$,
\begin{eqnarray*}
s_t^m \sigma (L y) &=&s_t^m \sigma (L^{(1)}L^{(2)} y)
\ =\ \sum_{r=0}^{m-1} L_{t,r}^{(1)}\, s_r^m \sigma (L^{(2)}y) \\
 &=& \sum_{r=0}^{m-1} L_{t,r}^{(1)}\sum_{j=0}^{m-1} L_{r,j}^{(2)}\,  s_j^m \sigma y 
\ = \sum_{j=0}^{m-1} \left(\sum_{r=0}^{m-1} L_{t,r}^{(1)} L_{r,j}^{(2)}\right) s_j^m \sigma y \\
&=& \sum_{j=0}^{m-1} \left(\left[\cR L^{(1)}\right]\!\!\left[\cR L^{(2)}\right]\right)_{t,j} s_j^m \sigma y.
\end{eqnarray*}
Hence
\begin{eqnarray*}
\sum_{j=0}^{m-1}\left[\cR L\right]_{t,j} s_j^m \sigma y\ = \sum_{j=0}^{m-1} \left(\left[\cR L^{(1)}\right]\!\!\left[\cR L^{(2)}\right]\right)_{t,j} s_j^m \sigma y
\end{eqnarray*}
for any $y \in \K[[\cB]]$, which implies the claim. 
\end{proof} 
\medskip

It follows that to compute $[\cR L]$ for an arbitrary operator $L \in \K[x]\langle E\rangle$, it suffices to apply the substitution
\begin{equation}
\label{substitute}
\begin{array}{llc}
E & \mapsto & [\cR E], \\
x & \mapsto & [\cR X], \\
1 & \mapsto & I_m
\end{array}
\end{equation}
where $I_m$ is the $m \times m$ identity matrix, to all terms of $L$. We adapt procedure \textsc{AssociatedOp} from p.~\pageref{AssOp} to compute the associated matrix of operators $[\cR L]$ for an $m$-sieved basis $\cB$ in the following way:

\medskip
\begin{center}
\large
\textbf{Procedure} \textsc{AssociatedOpSieved}
\end{center}
\label{AssocOpSieved}

\smallskip
\textsc{Input:} \verb+   + $L \in \K[x]\langle E\rangle$;\ $A \in \N$;

\verb+          + an $m$-sieved factorial basis $\cB = \langle P_{mk+j}(x)\rangle_{m\in \N \setminus\{0\},\ j \in \{0,1,\dots,m-1\}}$,\\[-10pt]

\verb+                                   + $\,(A,0)$-compatible with $E$

\medskip
\textsc{Output:} \verb+ + $\left[\cR L\right]\, =\, \left[L_{r,j}\right]_{r,j=0}^{m-1}\ \in\  M_m(\cE)$

\vskip 1pc
\begin{enumerate}
\item Using linear algebra, compute $\alpha_{k,j,i} \in \K$ for $j \in \{0,1,\dots,m-1\}$ and $i \in \{-A,-A+1, \dots,0\}\,$ such that\\[5pt]
\verb+         +$P_{mk+j}(x+1)\ \ = \ \displaystyle\sum_{i=-A}^0 \alpha_{k,j,i} P_{mk+j+i}(x)$\ \ \ for all $k\in\N$.\\[-1pt]

For $r, j \in \{0,1,\dots,m-1\}$\, let $\ E_{r,j}\ \, = \!\!\!\!\!\!\displaystyle\sum_{\myatop{\,-A \le i \le 0 }{\, i+j \equiv r \pmod{\!m}}}\!\!\!\!\!\alpha_{k+\frac{r-i-j}{m},j,i} S^{\frac{r-i-j}{m}}\, \in\, \cE.$\\

Let $[\cR E] =\left[E_{r,j}\right]_{r,j=0}^{m-1}$.

\item Using linear algebra, compute $\beta_{k,j,0},\, \beta_{k,j,1} \in \K$ for $j \in \{0,1,\dots,m-1\}$ such that\\[5pt]
\verb+         +$x P_{mk+j}(x)\, =\, \beta_{k,j,0}P_{mk+j}(x) + \beta_{k,j,1}P_{mk+j+1}(x)$\ \ for all $k\in\N$.\\[-3pt]

For $r, j \in \{0,1,\dots,m-1\}$ let\, $\ X_{r,j}\ \, = \!\!\!\!\!\!\displaystyle\sum_{\myatop{\,0 \le i \le 1 }{\, i+j \equiv r \!\!\pmod{\!m}}}\!\!\!\!\!\beta_{k+\frac{r-i-j}{m},j,i} S^{\frac{r-i-j}{m}}\, \in\, \cE.$\\

Let  $[\cR X] =\left[X_{r,j}\right]_{r,j=0}^{m-1}$.

\item Return the matrix of operators $[\cR L] =\left[L_{r,j}\right]_{r,j=0}^{m-1}$, obtained by applying substitution (\ref{substitute}) to $L$.
\end{enumerate}

\medskip
For a product binomial-coefficient basis $\cC_{\mathbf{a},\mathbf{b}}$, we can make the above procedure 
more specific, as already explained in part on p.~\pageref{expansion} ff.

\begin{proposition}
\label{xrjcab}
Let $\cB = \cC_{\mathbf{a},\mathbf{b}}$. Then for all $r,j \in \{0,1,\ldots,m-1\}$,
\[
X_{r,j}\ =\ [r = j]\,\frac{k-b_{j+1}}{a_{j+1}} +
[r = 0 \land j = m - 1]\,\frac{k}{a_{j+1}}\,S^{-1} +
[r = j+1]\,\frac{k+1}{a_{j+1}}
\]
where 
\[
[\varphi]\ =\ \left\{
\begin{array}{ll}
1, & {\ if\ } \varphi {\ is\ true}, \\
0, & {\ otherwise}
\end{array}
\right.
\]
is the Iverson bracket.
\end{proposition}

\begin{proof}
From Proposition \ref{x} we read off that in this case
\begin{align}
\beta_{k,j,0} &\ =\ \frac{k-b_{j+1}}{a_{j+1}}, \label{kj0} \\
\beta_{k,j,1} &\ =\ \frac{k+1}{a_{j+1}}. \label{kj1}
\end{align}
From (\ref{Lrjdef}) with $A = 0$, $B = 1$ it follows that
\begin{align}
X_{r,j} &\ =\ [j\equiv r \!\!\!\!\!\pmod{m}]\,\beta_{k+\frac{r-j}{m},j,0} S^{\frac{r-j}{m}}\nonumber  \\
&\ +\ [j\equiv r-1 \!\!\!\!\!\pmod{m}]\,\beta_{k+\frac{r-j-1}{m},j,1} S^{\frac{r-j-1}{m}}. \label{xrjm}
\end{align}
Combining (\ref{kj0}) -- (\ref{xrjm}) with $0 \le r,j \le m-1$ yields the assertion. 
\end{proof}

\smallskip
\begin{center}
\large
\textbf{Algorithm} \textsc{AssociatedOpBC}
\end{center}
\label{AssocOpBC}

\smallskip
\textsc{Input:} \verb+   + $L \in \K[x]\langle E\rangle$, $m \in \N\setminus\{0\}$,

\verb+          + $\mathbf{a} = (a_1,a_2,\ldots,a_m) \in  (\N\setminus\{0\})^m$,\  $\mathbf{b} = (b_1,b_2,\ldots,b_m) \in \Z^m$

\medskip
\textsc{Output:} \verb+ + $\left[\cR L\right]\, =\, \left[L_{r,j}\right]_{r,j=0}^{m-1}$\ \ where $\cB = \cC_{\mathbf{a},\mathbf{b}}$

\vskip 1pc
\begin{enumerate}
\item $A := m \max_{1 \le i \le m} a_i$.
\item For $j = 0,1, \ldots,m-1$ let
\[
P^{(\pi)}_{mk+j}(x) \ :=\ \prod_{i=1}^j \binom{a_i x + b_i}{k+1} \cdot \prod_{i=j+1}^m \binom{a_i x + b_i}{k}.
\]
\item For $j = 0,1,\ldots,m-1$ do \\
\phantom{quad} for $i = -A, -A+1, \ldots, 0$ compute by simplification
\begin{eqnarray*}
Q_{mk+j+i}(x) &:=& \frac{P^{(\pi)}_{mk+j+i}(x)}{P^{(\pi)}_{mk+j-A}(x)}\ \in\ \K(k)[x], \\
R_{mk+j}(x) &:=& \frac{P^{(\pi)}_{mk+j}(x+1)}{P^{(\pi)}_{mk+j-A}(x)}\ \in\ \K(k)[x].
\end{eqnarray*}
\phantom{quad} Equate coefficients of $1, x, x^2, \ldots, x^A$ on both sides of
\[
R_{mk+j}(x)\ =\ \sum_{i = -A}^0 \alpha_{k,j,i} Q_{mk+j+i}(x)
\]
\phantom{quad} and solve the resulting system of $A+1$ linear algebraic equations\\
\phantom{quad} for the $A+1$ unknowns $\alpha_{k,j,i} \in \K(k)$, $i=-A, -A+1, \ldots, 0$.
\item For $r,j = 0,1,\ldots,m-1$ let
\begin{align*}
E_{r,j} &:= \!\!\!\!\!\sum_{\myatop{\,-A \le i \le 0 }{\, i+j \equiv r \!\!\!\!\!\pmod{\!m}}}\!\!\!\!\!\alpha_{k+\frac{r-i-j}{m},j,i} S^{\frac{r-i-j}{m}},\\
X_{r,j} &:= [r = j]\frac{k-b_{j+1}}{a_{j+1}} +
[r = 0 \land j = m - 1]\frac{k}{a_{j+1}}\,S^{-1} +
[r = j+1]\frac{k+1}{a_{j+1}}.
\end{align*}
Let $[\cR E] =\left[E_{r,j}\right]_{r,j=0}^{m-1}$,\ \ $[\cR X] =\left[X_{r,j}\right]_{r,j=0}^{m-1}$.

\item Return the matrix of operators $[\cR L] =\left[L_{r,j}\right]_{r,j=0}^{m-1}$, obtained by applying substitution (\ref{substitute}) to $L$.

\end{enumerate}

\section{Main examples}
\label{sec:ex}

To find definite-sum solutions $y\in\K^\N$ of a recurrence equation of the form $L y = p$ where $L \in \K[n]\langle E\rangle$ and $p \in \K[n]$, we select a quasi-triangular, shift-compatible, $m$-sieved factorial basis $\cB$, and use Procedure \textsc{DefiniteSumSols} on p.~\pageref{defsumsols}. First, we follow Procedure \textsc{AssociatedOpSieved} on p.~\pageref{AssocOpSieved} (or, if $\cB$ is a product-binomial coefficient basis, Algorithm \textsc{AssociatedOpBC} on p.~\pageref{AssocOpBC}) to compute the matrix of operators $[\cR L] =\left[L_{r,j}\right]_{r,j=0}^{m-1}$. Then we set up the system of linear recurrence equations
\begin{align}
L_{0,0}  s_0^m c + L_{0,1}  s_1^m c + \dots + L_{0,m-1} s_{m-1}^m c \ &=\ s_0^m d\nonumber\\
L_{1,0}  s_0^m c + L_{1,1}  s_1^m c + \dots + L_{1,m-1} s_{m-1}^m c \ &=\ s_1^m d\label{system}\\
&\hspace*{5pt}\vdots\nonumber\\
L_{m-1,0}  s_0^m c + L_{m-1,1}  s_1^m c + \dots + L_{m-1,m-1} s_{m-1}^m c \ &=\ s_{m-1}^m d \nonumber
\end{align}
for the unknown sequence $c$ where $d = \langle d_0, d_1, \dots, d_{\deg p},0,0,\dots\rangle$ is the sequence of coefficients of polynomial $p(n)$. From Corollary \ref{cor} it follows that $y\in\K^\N$ with $y_n = \sum_{k=0}^{f(n)} c_k P_k(n)$ and $f(n)$ as in Definition~\ref{qt} satisfies $Ly = p$ if and only if the $m$-sections $s_j^m c$ of the coefficient sequence $c$ of $y$ satisfy (\ref{system}).  So finally we solve this system for the unknown sequence $c$ in any way we can. 

\begin{example}
\label{ex:main}
We illustrate the process just described on the linear recurrence equation $Ly = 0$ where $L\in\Q[n]\langle E\rangle$ is the $7^\textrm{th}$-order operator
\begin{align*}
L\ &=\ (n+8) (27034107689\, n+247037440535)\, E^7\\
&-2 (n+7) (27034107689\, n^2+707256640479\, n+3519513987204)\, E^6\\
&+(27034107689\, n^4+1763504948043\, n^3+29534526868562\, n^2\\
&\qquad\qquad\qquad +187161930754966\, n+404930820118700)\, E^5\\
&-4 (121973169216\, n^4+3928755304511\, n^3+43197821249228\, n^2\\
&\qquad\qquad\qquad +198945697078905\, n+329021406797184)\, E^4\\
&+(2167208392754\, n^4+45326791213914\, n^3+347739537911929\, n^2\\
&\qquad\qquad\qquad +1165212776491303\, n+1439937061155596)\, E^3\\
&-2 (613023852648\, n^4+8954947813901\, n^3+52565810509778\, n^2\\
&\qquad\qquad\qquad +141274453841469\, n+142893654078876)\, E^2\\
&-(n+2)^2 (1109455476579\, n^2+3624719391913\, n-357803625948)\, E\\
&+24 (n+1)^2 (n+2) (8996538731\, n+29816968829),
\end{align*}
using the basis $\cB = \cC_{(1,1),(0,0)}$ from Example \ref{xk2E} with
\[
P_{2k}(n) =  \binom{n}{k}^2, \quad P_{2k+1}(n) =  \binom{n}{k+1}  \binom{n}{k}.
\]
To compute $[\cR E]$, comparing (\ref{LPskj}) with (\ref{p2k}) and (\ref{p2k1}) yields
\[
\begin{array}{ccccccc}
\alpha_{k,0,0} & = & 1 & \qquad & \alpha_{k,1,0} & = & 1 \\[2pt]
\alpha_{k,0,-1} & = & 2 & \qquad & \alpha_{k,1,-1} & = & \frac{2k+1}{k+1} \\[3pt]
\alpha_{k,0,-2} & = & 1 & \qquad & \alpha_{k,1,-2} & = & \frac{k}{k+1},
\end{array}
\]
hence by (\ref{Lrjdef})
\begin{equation}
\label{rbe}
\left[\cR E \right]\ =\ \left[
\begin{array}{cc}
E_{0,0} & E_{0,1} \\
E_{1,0} & E_{1,1}
\end{array}
\right]\ =\ \left[
\begin{array}{cc}
S +1 & \frac{2k+1}{k+1} \\
2 S & \frac{k+1}{k+2} S + 1
\end{array}
\right],
\end{equation}
while from Proposition \ref{xrjcab} it follows that
\begin{equation}
\label{rbx}
\left[\cR X \right]\ =\ \left[
\begin{array}{cc}
X_{0,0} & X_{0,1} \\
X_{1,0} & X_{1,1}
\end{array}
\right]\ =\ \left[
\begin{array}{cc}
k & k S^{-1} \\
k+1 & k
\end{array}
\right].
\end{equation}
To obtain $[\cR L]=[L_{r,j}]_{r,j=0}^1\in M_2(\cE)$, we take $L$ and substitute $[\cR E]$ for $E$, $[\cR X]$ for $n$ and the $2\times 2$ identity matrix, multiplied by $z$, for any constant $z\in\K$. This yields the operators $L_{0,0},L_{0,1},L_{1,0},L_{1,1}$ given in~\ref{app:ex:main}.

The next step is to find a non-zero solution of the system of recurrences (\ref{system}), which in the case $m=2$ turns into
\begin{align}
L_{0,0}  s_0^2 c + L_{0,1}  s_1^2 c \ &=\ 0, \nonumber\\
L_{1,0}  s_0^2 c + L_{1,1}  s_1^2 c \ &=\ 0. \label{syssol}
\end{align}
By means of any algorithm for finding hypergeometric solutions of recurrence equations, we discover that the sequence $\langle k!\rangle_{k=0}^\infty$ is annihilated both by $L_{0,0}$ and $L_{1,0}$, while the sequence $\langle 2^k\rangle_{k=0}^\infty$ is annihilated both by $L_{0,1}$ and $L_{1,1}$. Hence the pairs $\langle k!, 0\rangle$ and $\langle 0,2^k\rangle$ are two linearly independent
solutions of (\ref{syssol}), and the interlacings\footnote{see Definition~\ref{def1}}
\[
c^{(1)}\ =\ \Lambda(k!, 0), \qquad c^{(2)}\ =\ \Lambda(0, 2^k)
\]
give rise to \textbf{\emph{two linearly independent definite-sum solutions}} of $Ly = 0$:
\begin{align*}
y^{(1)}_n\ &=\ \sum_{k=0}^{2n} c^{(1)}_k P_k(n)\ =\ \sum_{j=0}^{n} j! P_{2j}(n)\ =\ \sum_{j=0}^{n} j! \binom{n}{j}^2,\\
y^{(2)}_n\ &=\ \sum_{k=1}^{2n-1} c^{(2)}_k P_k(n)\ =\ \sum_{j=0}^{n-1} 2^j P_{2j+1}(n)\ =\ \sum_{j=0}^{n-1} 2^j \binom{n}{j+1}\binom{n}{j}. 
\end{align*}
These solutions can be used to factor the operator $L$ in $\Q(n)\langle E\rangle$. 
Zeilberger's algorithm \cite{Zeil90, Zeil91} computes operators 
\begin{align*}
L_1\ &=\ E^2 - 2 (n+2) E + (n+1)^2 \\
L_2\ &=\  (n+1) (n+3) E^2  - 3 (n+2) (2 n+3) E + (n+1) (n+2)
\end{align*}
such that $L_1 y^{(1)} = 0$ and $L_2 y^{(2)} = 0$. Algorithm Hyper \cite{hyper} shows that $L_1$ and $L_2$ are minimal annihilators of $y^{(1)}$ resp.\ $y^{(2)}$, hence they divide $L$ from the right, and so does their least common left multiple
\begin{align*}
L_4\ &=\ (5 + n) (-2 - 37 n - 138 n^2 - 123 n^3 - 33 n^4 + n^5 + n^6) E^4\\
&- (4 + n) (-276 - 1434 n - 2946 n^2 - 2342 n^3 - 718 n^4 - 35 n^5 + 18 n^6 + 2 n^7) E^3\\
&+ (-6896 - 32704 n - 60998 n^2 - 55528 n^3 - 26184 n^4 - 5888 n^5 - 239 n^6\\
&\qquad\qquad\qquad\qquad + 144 n^7 + 24 n^8 + n^9) E^2\\
&-(2+n)^2 (-1686 - 6102 n - 8388 n^2 - 5286 n^3 - 1430 n^4 - 44 n^5 + 47 n^6\\
&\qquad\qquad\qquad\qquad + 6 n^7) E \\ 
&+(1 + n)^2 (2 + n) (-331 - 803 n - 680 n^2 - 225 n^3 - 13 n^4 + 7 n^5 + n^6).
\end{align*}
Indeed, in $\Q(n)\langle E\rangle$, $L = L_3 L_4$ where 
\begin{align*}
L_3\ &=\ \frac{1}{p(n)}\cdot 24 (29816968829 + 8996538731\, n)\\
&+ \frac{1}{p(n)q(n)}\cdot (-441648942295148 - 1125518881632823\, n\\
& - 1053315627513055 n^2 - 421766464222932 n^3 - 59182885147037 n^4\\
& + 6062805507491 n^5 + 2492693169923 n^6 + 186046100685 n^7) E\\
&- \frac{1}{q(n) r(n)}\cdot (-9439737938111061 - 12584340359048430\, n\\
&- 5882475183305081 n^2 - 948888850906974 n^3 + 102412128495705 n^4\\
&+ 57283235096898 n^5 + 7386765251173 n^6 + 325688030730 n^7) E^2\\
&+ \frac{1}{r(n)}\cdot (247037440535 + 27034107689\, n) E^3
\end{align*}
with
\begin{align}
p(n)\ &=\ -331 - 803 n - 680 n^2 - 225 n^3 - 13 n^4 + 7 n^5 + n^6,\label{p(n)}\\
q(n)\ &=\ -2044 - 2849 n - 1348 n^2 - 187 n^3 + 37 n^4 + 13 n^5 + n^6,\nonumber\\
r(n)\ &=\ -6377 - 5887 n - 1542 n^2 + 111 n^3 + 117 n^4 + 19 n^5 + n^6.\nonumber
\end{align}
This means that any solution $y$ of $L_1 y = 0$ or $L_2 y = 0$ also satisfies our original equation $Ly=0$. By using the well-known method of variation of constants (a.k.a.\ reduction of order) we find out that whenever $c_2(n) a_{n+2} + c_1(n) a_{n+1} + c_0(n) a_{n} = 0$ for some $c_0, c_1, c_2 \in \Q(n)$ and for all $n\in \N$, then also $c_2(n) b_{n+2} + c_1(n) b_{n+1} + c_0(n) b_{n} = 0$ where
\begin{align}
\label{reduce}
b_n\ =\ a_n \sum_{k=0}^{n-1}(-1)^k \prod_{i=1}^k \left(1+\frac{c_1(i-1)}{c_2(i-1)}\frac{a_i}{a_{i+1}}\right),
\end{align}
provided that it is well defined. Thus from $L_1 y^{(1)} = 0$ we obtain another solution
\[
y_n^{(3)}\ =\ y_n^{(1)} \sum_{k=0}^{n-1} (-1)^k \prod_{i=1}^k \left(1 - 2 (1 + i) \frac{y_i^{(1)}}{y_{i+1}^{(1)}}\right)
\]
satisfying $L_1 y^{(3)} = 0$ and linearly independent from $y^{(1)}$. If we now use formula (\ref{reduce}) to construct a sequence $y^{(4)}$, linearly independent from $y^{(2)}$ and satisfying  $L_2 y^{(4)} = 0$, we unfortunately obtain the same solution $y^{(2)}$ again (but, in compensation, discover a summation identity). However, increasing the lower bound in the product within (\ref{reduce}) by 1, we do obtain a solution
\[
y_n^{(4)}\ =\ y_n^{(2)} \sum_{k=0}^{n-1} (-1)^k \prod_{i=2}^k \left(1 - 3 \frac{(1 + i) (1 + 2 i)}{i (2 + i)} \frac{y_i^{(2)}}{y_{i+1}^{(2)}}\right)
\]
satisfying $L_2 y_n^{(4)} = 0$ at all $n \in \N$ except at $n=0$. It is easy to check that $y^{(1)}$, $y^{(2)}$, $y^{(3)}$, $y^{(4)}$ are four linearly independent sequences, represented explicitly and satisfying $Ly_n = 0$ for all $n\ge 1$.

Since $L_4 = \tilde L_1 L_1 = \tilde L_2 L_2$ where
\begin{align*}
\tilde L_1\ &=\ (5 + n) (-2 - 37 n - 138 n^2 - 123 n^3 - 33 n^4 + n^5 + n^6) E^2\\
& - (4 + n) (-256 - 1060 n - 1492 n^2 - 836 n^3 - 142 n^4 + 21 n^5 + 6 n^6) E\\ 
& + (2 + n)\, p(n)\\
\tilde L_2\ &=\ \frac{1}{3 + n}(-2 - 37 n - 138 n^2 - 123 n^3 - 33 n^4 + n^5 + n^6) E^2\\
& - \frac{1}{3 + n}(-393 - 1698 n - 2727 n^2 - 1917 n^3 - 574 n^4 - 36 n^5 + 14 n^6 + 2 n^7) E\\
& + (1 + n)\,p(n)
\end{align*}
with $p(n)$ as given in (\ref{p(n)}), it follows that $L$ factorizes as
\begin{align}
\label{factorization}
L\ =\ L_3 \tilde L_1 L_1\ =\ L_3 \tilde L_2 L_2.
\end{align}
Since $\ord L_3 = 3$, $\ord L_1 = \ord L_2 = \ord \tilde L_1 = \ord \tilde L_2 = 2$, and the operators $L_3$, $L_1$, $L_2$, $\tilde L_1$, $\tilde L_2$ have neither right nor left first-order factors in $\Q(n)\langle E\rangle$, equation (\ref{factorization}) gives two distinct factorizations of $L$ into \emph{irreducible factors} in $\Q(n)\langle E\rangle$. 
\end{example}

When we are interested in finding $y \in \ker L$ of the form 
\[
y(x)\ =\ 
\sum_{k=0}^\infty c_k P_{mk}(x)
\]
for some $m$-sieved basis $\cB = \langle P_k(x)\rangle_{k=0}^\infty$,
we have $s_0^m \sigma_\cB y = c$ and $s_j^m \sigma_\cB y = 0$ for all $j \in \{1,2,\dots,m-1\}$, hence Corollary \ref{cor} implies
\begin{eqnarray}
L y = 0 &\Longleftrightarrow& \forall r \in \{0,1,\ldots,m-1\}\!: L_{r,0}\,  s_0^m \sigma_\cB y\ =\ 0 \nonumber \\
&\Longleftrightarrow&  \forall r \in \{0,1,\ldots,m-1\}\!: L_{r,0}\, c\ =\ 0 \nonumber \\
&\Longleftrightarrow&  { \gcrd}(L_{0,0}, L_{1,0},\ldots, L_{m-1,0})\, c\ =\ 0. \nonumber
\end{eqnarray}
This means that any nonzero element of the first column of $[\cR L] = \left[L_{r,j}\right]_{r,j=0}^{m-1}$ may serve as a nontrivial annihilator $L'$ of $h$, and taking their greatest common right divisor might yield $L'$ of lower order. The fact that we only need the first column $[\cR L]e^{(1)}$ of $[\cR L]$ -- where $e^{(1)} = (1,0,\ldots,0)^{ T} \in \K^m$ -- can be used to advantage in the following way:

In step 3 of Procedure \textsc{AssociatedOpSieved} from p.~\pageref{AssocOpSieved}, as well as in step 5 of Algorithm \textsc{AssociatedOpBC} from p.~\pageref{AssocOpBC}, instead of computing the entire matrix $[\cR L]$ we only compute the vector $[\cR L]e^{(1)}$. To do so, we start with $e^{(1)}$, and proceed from right to left through the expression obtained from $L$ by substitution (\ref{substitute}), multiplying a matrix by a vector at each point. Finally, we return $L' = { \gcrd}(L_{0,0}, L_{1,0},\ldots, L_{m-1,0})$.


\begin{example}[Ap\'ery's $\zeta(2)$-recurrence \cite{Apery, Carsten07}]\label{exm:apery2}

Let 
\begin{equation}\label{equ:Apery2}
L  := (n + 2)^2 E^2 - (11 n^2 + 33 n + 25) E - (n + 1)^2
\end{equation}
and 
\[
K(n,k) = \binom{n}{k} \binom{n+k}{2k}. 
\]
Take
\begin{eqnarray*}
P_{3k}(x) &=& K(x,k) \ =\ \binom{x}{k} \binom{x+k}{2k}, \\
P_{3k+1}(x) &=& \binom{x}{k} \binom{x+k}{2k+1}, \\
P_{3k+2}(x) &=& \binom{x}{k+1} \binom{x+k}{2k+1},
\end{eqnarray*}
or, more concisely,
\[
P_k(x) \ =\ \binom{x}{\lfloor\frac{k+1}{3}\rfloor} \binom{x+\lfloor\frac{k}{3}\rfloor}{\lfloor\frac{2k+1}{3}\rfloor}.
\]
Clearly $\cB = \langle P_k(x)\rangle_{k=0}^{\infty}$ is factorial, so $X$ is $(0,1)$-compatible with $\cB$: 
\begin{eqnarray*}
x P_{3 k}(x) &=& (2k+1) P_{3 k+1}(x) + k P_{3 k}(x),\\
x P_{3 k+1}(x) &=& (k+1) P_{3 k+2}(x) + k P_{3 k + 1}(x),\\ 
x P_{3 k+2}(x) &=& 2(k+1) P_{3 k+3}(x) - (k + 1) P_{3 k+2}(x).
\end{eqnarray*}
It is not hard to see that $\cB$ is quasi-triangular with $f(n) = 3n$.
Furthermore, the root sequence of $\cB$ 
\[
\rho =  \langle 0, 0, -1, 1, 1, -2, 2, 2, -3, 3, 3, -4, 4, 4, -5, 5, 5, -6, \dots \rangle
\]
satisfies
\[
[\rho_1 + 1, \rho_2 + 1,  \dots, \rho_k+ 1]\ \subseteq\ [\rho_1, \rho_2, \dots,\rho_k, \rho_{k+1}, \dots, \rho_{k+3}]
\]
for all $k \ge 0$, hence by Proposition \ref{compatroots}, $E$ is $(3,0)$-compatible with $\cB$:
\begin{align*}
P_{3 k}(x + 1)\ &=\ P_{3 k}(x)
+\ \frac{3}{2} P_{3 k - 1}(x) + \frac{8 k - 3}{2 k} P_{3 k - 2}(x) +  P_{3 k - 3}(x), \\[7pt]
P_{3 k+1}(x + 1)\ &=\ P_{3 k+1}(x)
+\ \frac{3 k + 1}{2 k + 1} P_{3 k}(x) + \frac{k}{2 k + 1} P_{3 k - 1}(x) + \frac{2 k - 1}{2 k + 1} P_{3 k - 2}(x), \\[7pt]
P_{3 k+2}(x + 1) \ &=\ P_{3 k+2}(x) + \frac{3 k + 2}{k + 1} P_{3 k+1}(x) + P_{3 k}(x).
\end{align*}
The associated operator matrices are:
\begin{eqnarray*}
\left[\cR X \right]&=& \left[
\begin{array}{ccc}
k & 0 & 2k\, S^{-1} \\
2k+1 & k & 0 \\
0 & k+1 & -(k+1)
\end{array}
\right],\\
\left[\cR E \right]&=& \left[
\begin{array}{ccc}
S +1 & \frac{3k+1}{2k+1} & 1 \\
\frac{8k+5}{2(k+1)} S & \frac{2k+1}{2k+3} S + 1 & \frac{3k+2}{k+1} \\
\frac{3}{2} S & \frac{k+1}{2k+3} S & 1
\end{array}
\right].
\end{eqnarray*}
\medskip
For $L$ as defined in~\eqref{equ:Apery2} we obtain:
\begin{align*}
L_{0,0} \ =\ &\ (k+2)^2 S^2  + \frac{29 k^3 + 46 k^2  + 14 k - 1}{2k+1} S \\
&\ \hspace*{-10pt} - 2(37 k^2 + 41 k + 11), \\[5pt]
L_{1,0} \ =\ &\ \frac{(k + 2) (4 k + 5) (12 k^2 + 26 k + 11)}{2 (k + 1) (2 k + 3)} S^2 \\
&\ \hspace*{-10pt} - \frac{79 + 237 k + 199 k^2 + 47 k^3}{2 (1 + k)} S - (2 k + 1) (49 k + 31), \\[5pt]
L_{2,0} \ =\ &\ \frac{(k+2) (22 k^2 + 62 k + 43)}{2 (2 k + 3)} S^2 \\
&\ \hspace*{-10pt} - \frac{3}{2} (11 k^2  + 34 k + 25) S - 11(k+1)(2 k + 1),
\end{align*}
and
\begin{equation}\label{eq:apery2:gcrd}
\gcrd(L_{0,0}, L_{1,0}, L_{2,0}) = S - 2\,\frac{2k+1}{k+1}.
\end{equation}
So $c_k = \binom{2k}{k}$ satisfies $L_{0,0} c_k = L_{1,0} c_k = L_{2,0} c_k =0$. 
Since
\begin{eqnarray*}
c_k P_{3k}(n)\ =\ \binom{2k}{k} \binom{n}{k}\binom{n+k}{2k}\ = \ \binom{n}{k}^2\binom{n+k}{k},
\end{eqnarray*}
we have found that Ap\'ery's $\zeta(2)$-sequence
\begin{eqnarray*}
y^{(1)}_n &=& \sum_{k=0}^\infty c_k P_{3k}(n)\ =\ \sum_{k=0}^n \binom{n}{k}^2\binom{n+k}{k}
\end{eqnarray*}
is a solution of $L y = 0$ with $L$ as in~\eqref{equ:Apery2}. 
Using formula (\ref{reduce}) as in Example~\ref{ex:main}, we obtain another, linearly independent solution $y^{(2)}$ of $Ly = 0$
\begin{eqnarray*}
y_n^{(2)}\ =\ y_n^{(1)} \sum_{k=0}^{n-1} (-1)^k \prod_{i=1}^k \left(1 - \frac{11 i^2 + 11 i + 3}{(i + 1)^2} \frac{y_i^{(1)}}{y_{i+1}^{(1)}}\right).
\end{eqnarray*}

\end{example}

\begin{example}[Ap\'ery's $\zeta(3)$-recurrence \cite{Apery, Carsten07}]\label{exm:apery3}

Let 
\begin{equation}\label{equ:Apery3}
L := (n + 2)^3 E^2 - (2n+3)(17 n^2 + 51 n + 39) E + (n + 1)^3
\end{equation}
and 
\[
K(n,k) = \binom{n+k}{2k}^2. 
\]
Take
\begin{eqnarray*}
P_{4k}(x) &=& K(x,k) \ =\ \binom{x+k}{2k}^2, \\
P_{4k+1}(x) &=& \binom{x+k}{2k} \binom{x+k}{2k+1}, \\
P_{4k+2}(x) &=& \binom{x+k}{2k+1}^2, \\
P_{4k+3}(x) &=& \binom{x+k}{2k+1} \binom{x+k+1}{2k+2}.
\end{eqnarray*}
Clearly $\cB = \langle P_k(x)\rangle_{k=0}^{\infty}$ is factorial, so $X$ is $(0,1)$-compatible with $\cB$: 
\begin{eqnarray*}
x P_{4 k}(x) &=& (2k+1) P_{4 k+1}(x) + k P_{4 k}(x),\\
x P_{4 k+1}(x) &=& (2k+1) P_{4 k+2}(x) + k P_{4 k + 1}(x),\\ 
x P_{4 k+2}(x) &=& 2(k+1) P_{4 k+3}(x) - (k + 1) P_{4 k+2}(x),\\
x P_{4 k+3}(x) &=& 2(k+1) P_{4 k+4}(x) - (k + 1) P_{4 k+3}(x).
\end{eqnarray*}
It is not hard to see that $\cB$ is quasi-triangular with $f(n) = 4n$.
Furthermore, the root sequence of $\cB$ 
\[
\rho =  \langle 0, 0, -1, -1, 1, 1, -2, -2, 2, 2, -3, -3, 3, 3, -4, -4, 4, 4, \dots \rangle
\]
satisfies
\[
[\rho_1 + 1, \rho_2 + 1,  \dots, \rho_k+ 1]\ \subseteq\ [\rho_1, \rho_2, \dots,\rho_k, \rho_{k+1}, \dots, \rho_{k+4}]
\]
for all $k \ge 0$, hence by Proposition \ref{compatroots}, $E$ is $(4,0)$-compatible with $\cB$:
\begin{align*}
P_{4 k}(x + 1)\ =&\ P_{4 k}(x) + 2 P_{4 k - 1}(x) + \frac{3k-1}{k} P_{4 k - 2}(x) +  \frac{4 k - 1}{k}P_{4 k - 3}(x)\\
&+\ P_{4 k - 4}(x), \\
P_{4 k+1}(x + 1)\ =&\ P_{4 k+1}(x)+ \frac{4 k + 1}{2 k + 1} P_{4 k}(x) + \frac{2k}{2 k + 1} P_{4 k - 1}(x)\\
&+\ \frac{2 k - 1}{2 k + 1}\left(P_{4 k - 2}(x) + P_{4 k - 3}(x)\right), \\
P_{4 k+2}(x + 1) \ =&\ P_{4 k+2}(x) + 2 P_{4 k+1}(x) + P_{4 k}(x),\\
P_{4 k+3}(x + 1) \ =&\ P_{4 k+3}(x) + \frac{4 k + 3}{2(k + 1)} P_{4 k+2}(x) + \frac{6 k + 5}{2(k + 1)}P_{4 k+1}(x) + P_{4k}(x).
\end{align*}
The associated operator matrices are:
\begin{eqnarray*}
\left[\cR X \right]&=& \left[
\begin{array}{cccc}
k & 0 & 0 & 2k\, S^{-1} \\
2k+1 & k & 0 & 0 \\
0 & 2k+1 & -(k+1) & 0 \\
0 & 0 & 2(k+1) & -(k+1)
\end{array}
\right],\\
\left[\cR E \right]&=& \left[
\begin{array}{cccc}
S +1 & \frac{4k+1}{2k+1} & 1 & 1 \\
\frac{4k+3}{k+1} S & \frac{2k+1}{2k+3} S + 1 & 2 & \frac{6k+5}{2(k+1)} \\
\frac{3k+2}{k+1} S & \frac{2k+1}{2k+3} S & 1 & \frac{4k+3}{2(k+1)} \\
2 S & \frac{2(k+1)}{2k+3} S & 0 & 1 
\end{array}
\right].
\end{eqnarray*}
\medskip
For $L$ as defined in \eqref{equ:Apery3} we obtain:
\begin{align*}
L_{0,0} \ =\ &\ (k+2)^3 S^2 + \frac{\left(58 k^4+105 k^3-25 k^2-121 k-45\right) S}{2
   k+1}\\
   &-4 (2 k+1) \left(90 k^2+101 k+27\right), \\[5pt]
L_{1,0} \ =\ &\ \frac{(k+2)^2 \left(28 k^3+96 k^2+103 k+34\right)
   S^2}{(k+1) (2 k+3)}\\
   &-\frac{2 \left(75 k^4+414 k^3+796
   k^2+636 k+177\right) S}{k+1}-8 (37 k+27) (2 k+1)^2, \\[5pt]
L_{2,0} \ =\ &\ \frac{(k+2)^2 \left(26 k^3+87 k^2+90 k+28\right) S^2}{(k+1)
   (2 k+3)}\\[5pt]
   &-\frac{4 \left(42 k^4+215 k^3+390 k^2+295
   k+77\right) S}{k+1}-16 (10 k+7) (2 k+1)^2,\\[5pt]
L_{3,0} \ =\ &\ \frac{2 (k+2)^2 \left(12 k^2+33 k+22\right) S^2}{2 k+3}-8
   \left(22 k^3+96 k^2+137 k+64\right) S\\[5pt]
   &-64 (k+1) (2 k+1)^2,
\end{align*}
and
\begin{equation}\label{eq:apery3:gcrd}
\gcrd(L_{0,0}, L_{1,0}, L_{2,0}, L_{3,0}) = S -\frac{4 (2 k+1)^2}{(k+1)^2}.
\end{equation}
So $c_k = \binom{2k}{k}^2$ satisfies $L_{0,0} c_k = L_{1,0} c_k = L_{2,0} c_k = L_{3,0} c_k = 0$. Since
\begin{eqnarray*}
c_k P_{4k}(n)\ =\ \binom{2k}{k}^2 \binom{n+k}{2k}^2\ = \ \binom{n}{k}^2\binom{n+k}{k}^2,
\end{eqnarray*}
we have found that Ap\'ery's $\zeta(3)$-sequence
\begin{eqnarray*}
y^{(1)}_n &=& \sum_{k=0}^{\infty} c_k P_{4k}(n)\ =\ \sum_{k=0}^{\infty} \binom{n}{k}^2\binom{n+k}{k}^2
\end{eqnarray*}
is a solution of $L y = 0$ with $L$ as in~\eqref{equ:Apery3}.
Using formula (\ref{reduce}) as in Example~\ref{ex:main}, we obtain another, linearly independent solution $y^{(2)}$ of $Ly = 0$
\begin{eqnarray*}
y_n^{(2)}\ =\ y_n^{(1)} \sum_{k=0}^{n-1} (-1)^k \prod_{i=1}^k \left(1 - \frac{(2i + 1) (17 i^2 + 17 i + 5)}{(i + 1)^3} \frac{y_i^{(1)}}{y_{i+1}^{(1)}}\right).
\end{eqnarray*}

\end{example}

\section{Shuffled polynomial bases}
\label{sec:gen}

In Examples~\ref{exm:apery2} and~\ref{exm:apery3}, we built two quasi-triangular bases that involved
some binomial coefficients. These bases, however, are sieved polynomial bases that cannot be 
written as a product basis of those binomial coefficients.

In the case of product bases (see Definition~\ref{def:pbasis}) the root sequences 
of the factors are interlaced in a balanced way. However, as illustrated in the previous examples, sometimes we need to 
interlace the root sequences in an unbalanced way. The root sequence of the basis
for Example~\ref{exm:apery2} was:
\[\rho = \langle 0, 0, -1, 1, 1, -2, 2, 2, -3, 3, 3, \ldots\rangle,\]
while the root sequences of the two factors $\binom{x}{k}$ and $\left\{\binom{x+k}{2k}, \binom{x+k}{2k+1}\right\}$ 
are, respectively:
\[\rho^{(1)} = \langle 0, 1, 2, 3, \ldots\rangle,\quad \rho^{(2)} = \langle 0, -1, 1, -2, 2, -3, 3, \ldots\rangle.\]

The usual interlacing of these two root sequences is different from the sequence $\rho$. However, it is possible to
build up the sequence $\rho$ from $\rho_1$ and $\rho_2$:
\[\begin{array}{rcl}
   \rho_{3k+1} &=& \rho^{(2)}_{2k+1},\\
   \rho_{3k+2} &=& \rho^{(1)}_{k+1},\\
   \rho_{3k+3} &=& \rho^{(2)}_{2k+2}.
\end{array}\]

It is this idea on which we base the concept of a \emph{shuffled polynomial basis}.
\begin{definition}\label{def:shuffle}
   For $F, m \in \N \setminus \{0\}$, let $\cB_i = \langle P_k^{(i)}\rangle_{k=0}^\infty$ for $i = 1, 2, \ldots, F$
   be polynomial bases, and let $\bc = (c_0, ..., c_{m-1}) \in \{1, ..., F\}^m$. For all
   $k \in \N$ and index $j \in \{0,\ldots, m-1\}$, let
   \begin{equation}\label{equ:shuffle}
      Q_{mk+j}(x) = \prod_{i=1}^{F} P_{e_i(mk+j)}^{(i)}(x),
   \end{equation}
   where $e_i(mk+j) = ks_i(m) + s_i(j)$, and $s_i(t) = |\{r \in \{0,\ldots,t-1\};\ c_r = i\}|$ for all $t \in \{0,\ldots,m\}$.
   Then the sequence $\langle Q_n(x)\rangle_{n=0}^\infty$ is the \emph{\textbf{c}-shuffled basis of $\cB_1,\ldots,\cB_F$}.
\end{definition}

Intuitively, $F$ is the number of bases we are shuffling, $m$ is the length of the period of this shuffling,
and the vector $\bc$ indicates which basis is used in each step for each period. Next, the function $e_i(n)$
gives the index of the element of the $i$-th basis that is used in the $n$-th element of the shuffled basis. Finally, the 
function $s_i: \{0,\ldots,m\} \rightarrow \N$ counts how many times we have used the $i$-th basis
in a particular point of a period.

\begin{example}
   Let $\cB_1,\ldots, \cB_m$ be factorial bases. Then $\prod_{i=1}^m \cB_i$ is the $(1,\ldots, m)$-shuffled basis
   of $\cB_1,\ldots,\cB_m$. Namely, the number of factors coincides with the length of the period (i.e., $F = m$),
   and the vector $\bc = (1,2,\ldots, m)$ indicates in what order we interlace the bases.

   In this particular example, the functions $s_i(j)$ are the step functions:
   \[s_i(j) = \left\{\begin{array}{ll}
      0 & \text{if } j < i,\\
      1 & \text{otherwise} 
   \end{array}\right.\]
   which leads to the following formulas for $e_i(n)$:
   \[e_i(mk+j) = k \cdot 1 + s_i(j) = \left\{\begin{array}{ll}
      k & \text{if } j < i,\\
      k+1 & \text{otherwise}
   \end{array}\right.\]
   turning~\eqref{equ:shuffle} into~\eqref{equ:product}.
\end{example}

\begin{example}
   Let us show that the basis of Example~\ref{exm:apery2} is a shuffled basis using the following two bases:
   \[\cB_1 = \left\langle\binom{x}{n}\right\rangle_n^\infty,\quad
   \cB_2 = \left\langle\binom{x+n}{2n},\binom{x+n}{2n+1}\right\rangle_n^\infty.\]

   Take $F = 2$, $m = 3$, and $\bc = (2,1,2)$. Then we have that 
   \[s_1(0) = 0, s_1(1) = 0,\qquad s_1(2) = 1,\qquad s_1(3) = 1,\]
   \[s_2(0) = 0, s_2(1) = 1,\qquad s_2(2) = 1,\qquad s_2(3) = 2.\]

   This implies that the index functions $e_1(n)$ and $e_2(n)$ have the following values:
   \[e_1(3k + j) = k\cdot 1 + s_1(j) = \left\{\begin{array}{ll}
      k & \text{if }j = 0,\\
      k & \text{if }j = 1,\\
      k+1 & \text{if }j = 2,
   \end{array}\right.\]
   \[e_2(3k + j) = k\cdot 2 + s_2(j) = \left\{\begin{array}{ll}
      2k & \text{if }j = 0,\\
      2k+1 & \text{if }j = 1,\\
      2k+1 & \text{if }j = 2,
   \end{array}\right.\]
   Hence, we have that the $(2,1,2)$-shuffled basis of $\cB_1$ and $\cB_2$ is:
   \begin{eqnarray*}
      Q_{3k}(x) &=& P_{k}^{(1)}(x)P_{2k}^{(2)}(x) = \binom{x}{k}\binom{x+k}{2k}, \\
      Q_{3k+1}(x) &=& P_{k}^{(1)}(x)P_{2k+1}^{(2)}(x) = \binom{x}{k} \binom{x+k}{2k+1}, \\
      Q_{3k+2}(x) &=& P_{k+1}^{(1)}(x)P_{2k+1}^{(2)}(x) = \binom{x}{k+1} \binom{x+k}{2k+1},
   \end{eqnarray*}
   which are exactly the formulas displayed in Example~\ref{exm:apery2}.
\end{example}

Shuffled bases are a particular case of sieved polynomial bases, since the elements are defined modulo $m$.

\begin{definition}\label{def:sieved_compatible}
   Let $L \in {\cal L}_{\K[x]}$, $m \in \N\setminus\{0\}$, and let $\cB = \langle P_n(x)\rangle_{n=0}^{\infty}$ be a factorial basis of $\K[x]$. 
   We say that $L$ is \emph{$(A,B)$-compatible in $m$ sections with $\cB$} if there are $\alpha_{k,j,i} \in \K(k)$ such 
   that for all $k \in \N$ and $j \in \{0,\ldots,m-1\}$,
      \[L P_{mk+j}(x) = \sum_{i=-A}^B \alpha_{k,j,i} P_{mk+j+i}(x).\]
\end{definition}

Observe that in this definition we have used formula \eqref{LPskj}, but restricting ourselves to $\alpha_{k,j,i}$ 
being rational expressions in $k$. This means that being $(A,B)$-compatible with $\cB$ in the usual sense (see 
Definition~\ref{def:compat}) does not imply that the operator is $(A,B)$-compatible with $\cB$ in one section. 
We include this restriction for computational reasons. In general,
the coefficients $\alpha_{k,j,i}$ could belong to any computable ring of sequences closed under 
shift and dilation.

\begin{example}
   Let $\cB = \langle P_n(x)\rangle_n^{\infty}$ be the factorial basis from Example~\ref{xk2E}, defined
   by the formulas:
   \[P_{2n}(x) = \binom{x}{n}^2,\qquad P_{2n+1}(x) = \binom{x}{n+1}\binom{x}{n}.\]

   Equations~\eqref{p2k} and~\eqref{p2k1} show that $E$ is $(2,0)$-compatible in 
   two sections with $\cB$. However, $E$ is not $(2,0)$-compatible in one section with $\cB$. To 
   see that, it is enough to write a unified compatibility formula that holds for all $n \in \N$:
   \begin{equation*}
      P_n(x+1) = P_n(x) + 
         \frac{n}{\lceil n/2\rceil}P_{n-1}(x) + 
         \frac{\lfloor n/2 \rfloor}{\lceil n/2\rceil}P_{n-2}(x).
   \end{equation*}

   This formula fits Definition~\ref{def:compat}, but the coefficients of $P_{n-1}(x)$ and
   $P_{n-2}(x)$ are not rational in $n$, so $E$ is not $(2,0)$-compatible in one section with $\cB$.
\end{example}

Let $\cB_1, \cB_2, \ldots, \cB_F$ be factorial bases, compatible in sections with an 
endomorphism $L$ or with $X$, and let $\cB$ be their $\bc$-shuffled basis. The following 
lemmas show that then $\cB$ is also compatible in sections with $L$ (resp. with $X$). Moreover, the 
proofs of these lemmas show how to construct the coefficients $\alpha_{k,j,i} \in \K(k)$ for $cB$, 
given the corresponding coefficients for each of $\cB_1, \cB_2, \ldots, \cB_F$.

First, we show in Lemma~\ref{lem:sections} how to expand the number 
of sections of compatibility. This result holds for any polynomial basis. In Proposition~\ref{prop:fac_tri}, we show that 
a shuffled basis of quasi-triangular bases is again quasi-triangular. Then we proceed to extend the desired 
compatibilities to shuffled bases (Theorems~\ref{shuffle_X} and~\ref{shuffle_E}). These extensions make use
of Proposition~\ref{compat}, so they hold for factorial bases.

\begin{lemma}\label{lem:sections}
   Let $L \in {\cal L}_{\K[x]}$ and $\cB = \langle P_n(x)\rangle_{n=0}^{\infty}$ be a polynomial basis. If 
   $L$ is $(A,B)$-compatible in $m$ sections with $\cB$ then it is $(A,B)$-compatible in $tm$ sections 
   with $\cB$ for all $t \in \N\setminus\{0\}$.
\end{lemma}
\begin{proof}
   Let $k \in \N$ and $j \in \{0,\ldots, tm-1\}$. Write $j = j_0m + j_1$ with $j_1 \in \{0,\ldots m-1\}$. Then we have: 
   \begin{align*}
      L P_{(tm)k + j} = &\ L P_{(tk+j_0)m + j_1} = \sum_{i=-A}^B \alpha_{tk+j_0, j_1, i} P_{m(tk + j_0) + j_1 + i}(x) = \\
      = & \sum_{i=-A}^B \tilde{\alpha}_{k, j, i} P_{(tm)k + j + i}(x),
   \end{align*}
   proving the compatibility in $tm$ sections, with $\tilde{\alpha}_{k,j,i} = \alpha_{tk+j_0, j_1, i} \in \K(k)$. 
\end{proof}

\begin{proposition}\label{prop:fac_tri}
   Let $\cB = \langle Q_n(x)\rangle_{n=0}^{\infty}$ be the $\bc$-shuffled basis of 
   the factorial bases $\cB_i = \langle P_n^{(i)}(x)\rangle_{n=0}^{\infty}$ for $i=1,\ldots, F$. 
   Then $\cB$ is also factorial. Moreover, if all $\cB_i$ are 
   quasi-triangular, so is $\cB$.
\end{proposition}
\begin{proof}
   Let $n \in \N$ be written as $n = km + j$ with $j \in \{0,\ldots, m-1\}$. Using
   Definition~\ref{def:shuffle}, we have that 
   \[Q_n(x) = \prod_{i=1}^{F} P_{ks_i(m) + s_i(j)}^{(i)}(x)\]
   where 
   \[s_i(t) = |\{r \in \{0,1,\dots,t-1\};\ c_r = i\}|\]
   for all $t \in \{0,1,\dots,m\}$. We wish to show that the quotient $Q_{n+1}(x)/Q_n(x)$ is a 
   polynomial of degree 1. To this end, we distinguish two cases.
   
   \smallskip
   \textbf{a)} $j \in \{0,1,\dots,m-2\}$
   
   \smallskip
   \noindent
   Here $n+1 = k m + (j + 1)$ where $j+1 \in \{1,2,\dots,m-1\}$. Obviously
   \begin{align*}
   s_i(j) &= |\{r \in \{0,1,\dots,j-1\};\ c_r = i\}|, \text{ and}\\
   s_i(j+1) &= |\{r \in \{0,1,\dots,j-1,j\};\ c_r = i\}|\\
            &= \begin{cases}
                           s_i(j), & c_j \ne i\\
                           s_i(j)+1, & c_j = i,
                  \end{cases}
   \end{align*}                  
   hence the quotient
   \begin{align*}
   \frac{P_{k s_i(m) + s_i(j+1)}^{(i)}(x)}{P_{k s_i(m) + s_i(j)}^{(i)}(x)} &=
   \begin{cases}
            1, & c_j \ne i\\
            \frac{P_{k s_i(m) + s_i(j)+1}^{(i)}(x)}{P_{k s_i(m) + s_i(j)}^{(i)}(x)}, & c_j = i
   \end{cases}            
   \end{align*}
   is a polynomial in $x$ of degree at most 1, and so, since the basis $\cB_i$ is factorial, the quotient
   \begin{align}
   \label{Q_{bn+1}/Q_n}
   \frac{Q_{n+1}(x)}{Q_n(x)} &= \prod_{i=1}^{F} \frac{P_{k s_i(m) + s_i(j+1)}^{(i)}(x)}{P_{k s_i(m) + s_i(j)}^{(i)}(x)} = \frac{P_{k s_{c_j}(m) + s_{c_j}(j)+1}^{(c_j)}(x)}{P_{k s_{c_j}(m) + s_{c_j}(j)}^{(c_j)}(x)}
   \end{align}
   is a polynomial in $x$ of degree 1.
   
  \smallskip
  \textbf{b)} $j = m-1$ (hence $m = j+1$)

   \smallskip
   \noindent
   Here $n = km + (m-1)$ and $n+1 = (k+1)m + 0$. Obviously
   \begin{eqnarray}
      s_i(0) &=& |\{r \in \emptyset;\ c_r = i\}| = 0, \nonumber\\               
      s_i(m-1) &=& |\{r \in \{0,1,\dots,m-2\};\ c_r = i\}|, \text{ and} \nonumber\\ 
      s_i(m) &=& |\{r \in \{0,1,\dots,m-2, m-1\};\ c_r = i\}|\nonumber\\
               &=& \begin{cases}
                              s_i(m-1), & c_{m-1} \ne i \\
                              s_i(m-1)+1, & c_{m-1} = i, 
                     \end{cases}\label{s_i(m)_s_i(m-1)}
   \end{eqnarray}
   hence the quotient
   \begin{align*}
      \frac{P_{(k+1) s_i(m) + s_i(0)}^{(i)}(x)}{P_{k s_i(m) + s_i(m-1)}^{(i)}(x)} &=
      \begin{cases}
         \frac{P_{(k+1) s_i(m-1)}^{(i)}(x)}{P_{(k+1) s_i(m-1)}^{(i)}(x)}, & c_{m-1} \ne i\\
         \frac{P_{(k+1) (s_i(m-1)+1)}^{(i)}(x)}{P_{k (s_i(m-1)+1)+s_i(m-1)}^{(i)}(x)}, & c_{m-1} = i 
      \end{cases}\\
      &=
      \begin{cases}
         1, & c_{m-1} \ne i\\
         \frac{P_{(k+1) s_i(m-1)+k+1}^{(i)}(x)}{P_{(k+1) s_i(m-1)+k}^{(i)}(x)}, & c_{m-1} = i
      \end{cases}            
   \end{align*}
   is a polynomial in $x$ of degree at most 1, and the quotient
   \begin{align}
      \frac{Q_{n+1}(x)}{Q_n(x)} &= \prod_{i=1}^F \frac{P_{(k+1) s_i(m) + s_i(0)}^{(i)}(x)}{P_{k s_i(m) + s_i(m-1)}^{(i)}(x)} 
      = \frac{P_{(k+1) s_{c_{m-1}}(m-1)+k+1}^{({c_{m-1}})}(x)}{P_{(k+1) s_{c_{m-1}}(m-1)+k}^{({c_{m-1}})}(x)}
   \end{align}
   is again a polynomial in $x$ of degree 1, so we conclude that the $\bc$-shuffled basis $\cB$ is indeed factorial as claimed. 
   
   Note that in the case \textbf{b)} it follows from \eqref{s_i(m)_s_i(m-1)} that $s_{c_{m-1}}(m-1) + 1 = s_{c_{m-1}}(m)$, 
   and since $m-1 = j$, we have
   \begin{align*}
   \frac{Q_{n+1}(x)}{Q_n(x)} &=   \frac{P_{(k+1) s_{c_{m-1}}(m-1)+k+1}^{({c_{m-1}})}(x)}{P_{(k+1) s_{c_{m-1}}(m-1)+k}^{({c_{m-1}})}(x)} 
   = \frac{P_{k(s_{c_{m-1}}(m-1)+1)+s_{c_{m-1}}(m-1)+1}^{({c_{m-1}})}(x)}{P_{k(s_{c_{m-1}}(m-1)+1)+s_{c_{m-1}}(m-1)}^{({c_{m-1}})}(x)}\\
   &= \frac{P_{k s_{c_{m-1}}(m)+s_{c_{m-1}}(m-1)+1}^{({c_{m-1}})}(x)}{P_{k s_{c_{m-1}}(m)+s_{c_{m-1}}(m-1)}^{({c_{m-1}})}(x)}
   = \frac{P_{k s_{c_j}(m)+s_{c_j}(j)+1}^{({c_j})}(x)}{P_{k s_{c_j}(m)+s_{c_j}(j)}^{({c_j})}(x)},
   \end{align*}
   so we conclude that equation \eqref{Q_{bn+1}/Q_n} in fact holds for all $j \in \{0,1,\dots, m-1\}$.
   
   Now, if all the bases $\cB_i$ are quasi-triangular, we can use Proposition~\ref{prop:quasi} to prove that $\cB$ is 
   also quasi-triangular. First, it is clear from Definition~\ref{def:shuffle} that the root sequence $\rho$ of 
   $\cB$ contains the root sequences $\rho^{(i)}$ of $\cB_i$ for all $i \in \{1,2,\dots,F\}$, so 
   $\langle 0,1,2,\ldots\rangle$ is certainly a subsequence of $\rho$. 

  Next, it is also clear that $\rho$ is a kind of interlacing of $\rho^{(1)}, \rho^{(2)}, \dots, \rho^{(F)}$, in which the 
  relative order of terms originating from $\rho^{(i)}$ is preserved within $\rho$ for each $i \in \{1,2,\dots,F\}$. 
  Let $n\in \N$ be arbitrary, and let $\mu$ be the minimal index such that $\rho_\mu = n+1$. Then 
  $n+1 = \rho_\mu = \rho_\nu^{(i)}$ for some $i \in \{1,2,\dots,F\}$ and $\nu \in \N \setminus \{0\}$. 
  Since $\cB_i$ is quasi-triangular, there is $\nu' < \nu$ such that $\rho_{\nu'}^{(i)} = n$. Hence there is also 
  $\mu' < \mu$ such that $\rho_{\mu'} = n$, so the first appearance of $n$ in $\rho$ (which occurs at some 
  $\mu''\le\mu'$) precedes the first appearance of $n+1$ in $\rho$. As $n\in \N$ was arbitrary, $\cB$ is quasi-triangular.
\end{proof}

\begin{theorem}\label{shuffle_X}
   Let $m \in \N \setminus \{0\}$ and $\cB = \langle Q_n(x)\rangle_{n=0}^{\infty}$ be the $\bc$-shuffled basis of 
   $\cB_i = \langle P_n^{(i)}(x)\rangle_{n=0}^{\infty}$ for $i=1,\ldots, F$. If each 
   $\cB_i$ is $(0,1)$-compatible in $t_i$ sections with $X$ (the multiplication-by-$x$ operator),
   then $\cB$ is $(0,1)$-compatible in $mt$ sections with $X$ for any $t \in \N\setminus\{0\}$ such that 
   $t_i$ divides $ts_i(m)$ for each $i = 1,\ldots, F$.
\end{theorem}
\begin{proof}
   By hypothesis, for each basis $\cB_i$, $k\in \N$ and $j \in \{0,\ldots, t_i-1\}$ we have: 
   \begin{equation}\label{equ:lem:shuffle_X}
      xP_{kt_i+j}^{(i)}(x) = \alpha_{k,j,0}^{(i)}P_{kt_i+j}^{(i)}(x) + \alpha_{k,j,1}^{(i)}P_{kt_i+j+1}^{(i)}(x).
   \end{equation}

   Let $t \in \N$ be such that $t_i$ divides $ts_i(m)$ for each $i = 1,\ldots, F$. Let $k \in \N$ and $j \in \{0,\ldots,mt-1\}$. 
   Let us see how we can express $x Q_{kmt + j}(x)$ in terms of $Q_{kmt+j}(x)$ and $Q_{kmt+j+1}(x)$.

   Let $j = j_0m + j_1$ with $j_1 \in \{0,\ldots, m-1\}$. We then have $kmt + j = (kt + j_0)m + j_1$, meaning that the only 
   difference between the element $Q_{kmt+j}(x)$ and $Q_{kmt+j+1}(x)$ is in the $c_{j_1}$-th factor of the shuffle basis. 
   Let $c = c_{j_1}$. In the polynomial $Q_{kmt+j}(x)$, the index of the factor from $\cB_c$ is
   \[(kt + j_0)s_c(m) + s_c(j_1).\]
   Let $a_c = ts_c(m)/t_c$. We know that $a_c \in \N$ by definition of $t$. Hence:
   \[(kt + j_0)s_c(m) + s_c(j_1) = (a_ck + j_2)t_c + j_3,\]
   where $j_0s_c(m) + s_c(j_1) = j_2t_c + j_3$ with $j_3 \in \{0,\ldots,t_c-1\}$. 
   
   In other terms, the element from the basis $\cB_c$ that appears in $Q_{kmt + j}(x)$ is the element 
   $P_{(a_ck + j_2)t_c + j_3}^{(c)}(x)$ and we can use formula~\eqref{equ:lem:shuffle_X}:
   \begin{align*} 
      x\frac{Q_{kmt + j}(x)}{Q_{kmt+j+1}(x)} & = x \frac{P_{(a_ck + j_2)t_c + j_3}^{(c)}(x)}{P_{(a_ck + j_2)t_c + j_3 + 1}^{(c)}(x)} 
                                             = \frac{xP_{(a_ck + j_2)t_c + j_3}^{(c)}(x)}{P_{(a_ck + j_2)t_c + j_3 + 1}^{(c)}(x)} \\
                                             & = \frac{\alpha_{a_ck+j_2,j_3,0}^{(c)}P_{(a_ck + j_2)t_c + j_3}^{(c)}(x) + \alpha_{a_ck+j_2,j_3,1}^{(c)}P_{(a_ck + j_2)t_c + j_3 + 1}^{(c)}(x)}{P_{(a_ck + j_2)t_c + j_3 + 1}^{(c)}(x)} \\
                                             & = \alpha_{a_ck+j_2,j_3,0}^{(c)}\frac{P_{(a_ck + j_2)t_c + j_3}^{(c)}(x)}{P_{(a_ck + j_2)t_c + j_3 + 1}^{(c)}(x)} + 
                                             \alpha_{a_ck+j_2,j_3,1}^{(c)} \\
                                             & = \alpha_{a_ck+j_2,j_3,0}^{(c)}\frac{Q_{kmt + j}(x)}{Q_{kmt+j+1}(x)} + \alpha_{a_ck+j_2,j_3,1}^{(c)}.
   \end{align*}
   Multiplying the last equation by $Q_{kmt+j+1}(x)$ we get that
   \[xQ_{kmt + j}(x) = \alpha_{a_ck + j_2,j_3,0}^{(c)}Q_{kmt + j}(x) + \alpha_{a_ck+j_2,j_3,1}^{(c)}Q_{kmt+j+1}(x),\]
   which proves that $\cB$ is $(0,1)$-compatible in $tm$ sections with $X$ and provides a direct formula 
   for the compatibility coefficients in each section.
\end{proof}

It is interesting to remark that the condition on $t$ guarantees that the new compatibility
coefficients are rational functions of $k$. If we pick $t$ to be minimal with such property, then we have a minimal number 
of sections for the compatibility of $X$. It could be that $X$ is compatible with $\cB$ in fewer sections, but
that is not the general case and should be taken care of individually.

\begin{theorem}\label{shuffle_E}
   Let $L \in {\cal L}_{\K[x]}$ be an endomorphism, $m\in \N\setminus\{0\}$, and $\cB = \langle Q_n(x)\rangle_{n=0}^{\infty}$ 
   be the $\bc$-shuffled basis of $\cB_i = \langle P_n^{(i)}(x)\rangle_{n=0}^{\infty}$ with $i=1,\ldots, F$. If each 
   $\cB_i$ is $(A_i,B_i)$-compatible in $t_i$ sections with $L$ then $\cB$ is $(mA,B)$-compatible in $mt$ sections with $L$
   where: 
   \begin{itemize}
      \item $B = \min\{B_i;\ i \in \{1,\ldots,F\}\}$,
      \item $A = \max\{\lceil A_i/s_i(m)\rceil;\ i \in \{1,\ldots,F\}\}$,
      \item $t$ is a natural number such that $t_i$ divides $ts_i(m)$ for each $i = 1,\ldots, F$.
   \end{itemize}
\end{theorem}
\begin{proof}
   By the equivalence of Proposition~\ref{compat}, it is enough to check conditions \textbf{C1} and \textbf{C2}
   in each section. For \textbf{C1}, we can repeat the same proof as in Theorem~\ref{prod} to show that 
   \[\deg L Q_n(x) \leq n + \min\{B_i;\ i \in\{1,\ldots,F\}\},\]
   showing that $B$ was chosen correctly.

   Consider $t \in \N$ as defined in this theorem. Let us see that, for any $k \in \N$ and $j\in\{0,\ldots, mt-1\}$,
   the polynomial $Q_{kmt + j - mA}(x)$ divides $Q_{kmt+j}(x)$. Since $L$ is an endomorphism, we have:
   \begin{equation}\label{equ:lem:shuffle_E:1}
      L Q_{kmt + j}(x) = L \left(\prod_{i=1}^F P_{kts_i(m) +s_i(j)}^{(i)}(x)\right) = \prod_{i=1}^F \left(L P_{kts_i(m) +s_i(j)}^{(i)}(x)\right).
   \end{equation}

   On the other hand,
   \begin{equation}\label{equ:lem:shuffle_E:2}
      Q_{kmt + j - mA}(x) = Q_{(kt-A)m + j}(x) = \prod_{i=1}^F P_{(kt-A)s_i(m) + s_i(j)}^{(i)}(x).
   \end{equation}
   At this point, we only need to show that, for all $i=1,\ldots,F$, the polynomial $P_{(kt-A)s_i(m) + s_i(j)}^{(i)}(x)$ divides $L P_{kts_i(m) +s_i(j)}^{(i)}(x)$.
   Let $a_i = (ts_i(m))/t_i$, which is a natural number by the construction of $t$. We can 
   also write $s_i(j) = j_{i,0}t_i + j_{i,1}$ with $j_{i,1} \in \{0,\ldots, t_i-1\}$. Then
   \[kts_i(m) +s_i(j) = (ka_i + j_{i,0})t_i + j_{i,1}.\]
   Using the $(A_i,B_i)$-compatibility of $\cB_i$ with $L$, we have
   \[L P_{kts_i(m) +s_i(j)}^{(i)}(x) = L P_{(ka_i + j_{i,0})t_i + j_{i,1}}^{(i)}(x) = \sum_{l=-A_i}^{B_i} \alpha_{(ka_i + j_{i,0}), j_{i,1},l}^{(i)} P_{kts_i(m) +s_i(j) + l}^{(i)}(x).\]

   If we now show that $(kt-A)s_i(m) + s_i(j) \leq kts_i(m)+s_i(j) - A_i$ for all $i$, then using the fact that $\cB_i$ is a factorial basis, we get that
   \[P_{(kt-A)s_i(m) + s_i(j)}^{(i)}(x)\text{ divides }P_{kts_i(m) +s_i(j) + l}^{(i)}(x)\text{ for all }l=-A_i,\ldots, B_i,\]
   and, in particular, that $P_{(kt-A)s_i(m) + s_i(j)}^{(i)}(x)$ divides $L P_{kts_i(m) +s_i(j)}^{(i)}(x)$.

   But this is simple to prove using the construction of $A$. Since $A \geq A_i/s_i(m)$ for all $i=1,\ldots, F$, then we have $As_i(m) \geq A_i$. Hence,
   \[(kt-A)s_i(m) +s_i(j) = kts_i(m) + s_i(j) - As_i(m) \leq kts_i(m) + s_i(j) - A_i.\]

   In the rest of the proof, we analyze the quotient between~\eqref{equ:lem:shuffle_E:1} and~\eqref{equ:lem:shuffle_E:2} to show that $L$ is $(mA,B)$-compatible 
   with $\cB$ in $mt$ sections, and that all compatibility coefficients for $\cB$ are rational functions in $k$ for all the sections.

   We first need to define a set of polynomials for each of the bases $\cB_i$. Let $i=1,\ldots, F$, $j \in \{0,\ldots, t_i-1\}$, $s \in \N$ and $k \in \N$.
   Consider the following quotient:
   \[D_{j,s,k}^{(i)}(x) = \frac{L P_{km +j}^{(i)}(x)}{P_{km+j-(A_i+s)}^{(i)}(x)}.\]
   This is always a polynomial by Proposition~\ref{compat}. Moreover, the coefficients of these polynomials
   are rational functions in $k$, since the compatibility conditions for each basis $\cB_i$ are given by rational functions in $k$.

   Also, for $i \in \{1,\ldots, F\}$, consider the integers $b_i = As_i(m) - A_i$. These are always non-negative integers since we have taken $A$
   such that $As_i(m) \geq A_i$.

   Now, we can analyze the quotient between~\eqref{equ:lem:shuffle_E:1} and~\eqref{equ:lem:shuffle_E:2}:
   \begin{equation}\label{equ:lem:shuffle_E:3}
      \begin{array}{rl}
      \displaystyle\frac{L Q_{kmt + j}(x)}{Q_{kmt + j - mA}(x)} 
                                                   & = \displaystyle\prod_{i=1}^F \frac{L P_{kts_i(m) +s_i(j)}^{(i)}(x)}{P_{(kt-A)s_i(m) + s_i(j)}^{(i)}(x)} \\ & \\
                                                   & = \displaystyle\prod_{i=1}^F \frac{L P_{(ka_i + j_{i,0})t_i + j_{i,1}}^{(i)}(x)}{P_{(ka_i + j_{i,0})t_i + j_{i,1} - (A_i + b_i)}^{(i)}(x)}\\ & \\
                                                   & = \displaystyle\prod_{i=1}^F D_{j_{i,1}, b_i, ka_i + j_{i,0}}^{(i)}(x)
      \end{array}
   \end{equation}
   We can now follow the proof of Proposition~\ref{compat} to convert the coefficients of the quotient~\eqref{equ:lem:shuffle_E:3} to the compatibility 
   coefficients for $\cB$. Let $k \in \N$ and $j \in \{0,\ldots,mt-1\}$. Consider, for $l \in \N$, the following family of polynomials
   \[I_{k,j,l} = \frac{Q_{kmt+j-mA + l}(x)}{Q_{kmt+j-mA}(x)} \in \K[x].\]
   Since $\cB$ is a factorial basis, it is clear that, for any fixed values of $k$ and $j$, the set $\{I_{k,j,l}; l\in\N\}$ is a factorial 
   basis of $\K[x]$. Let $\beta_{k,j,l}$ be the coefficients of the quotient~\eqref{equ:lem:shuffle_E:3} in terms of this basis. Then we get
   \begin{equation}\label{equ:lem:shuffle_E:4}
      \frac{L Q_{kmt + j}(x)}{Q_{kmt + j - mA}(x)} = \sum_{l=0}^{mA+B} \beta_{k,j,l}I_{k,j,l} = \frac{\sum_{l=0}^{mA+B}\beta_{k,j,l}Q_{kmt + j - mA + l}(x)}{Q_{kmt + j - mA}(x)}.
   \end{equation}

   The coefficients $\beta_{k,j,l}$ can be computed using linear algebra, performing a change of basis from the standard power basis $\{x^n; n\in \N\}$ to the
   basis created by $\{I_{k,j,l}; l\in \N\}$. Moreover, the upper bound for the sum in~\eqref{equ:lem:shuffle_E:4} is at most $mA+B$ since we already know that
   $L$ is $(mA,B)$-compatible with $\cB$.

   If we multiply both sides of equation~\eqref{equ:lem:shuffle_E:4} by $Q_{kmt + j - mA}(x)$ we obtain:
   \[L Q_{kmt + j}(x) = \sum_{l=0}^{mA+B} \beta_{k,j,l} Q_{kmt + j - mA + l}(x) = \sum_{i=-mA}^B \beta_{k,j,i+mA}Q_{kmt + j + i}(x).\]
   This formula provides the compatibility coefficients for the $(mA,B)$-compatibility of $L$ with $\cB$ in $mt$ sections, by taking 
   $\alpha_{k,j,i} = \beta_{k,j,i+mA}$. These coefficients are always rational functions in $k$ since $\cB$ is a shuffled basis of length $m$.
\end{proof}

This theorem is a direct generalization of Theorem~\ref{prod}, since for a product basis we have $s_i(m) = 1$ for 
each factor. Hence, in this case the definition of $A$ in Theorem~\ref{prod} coincides with the definition of $A$ in this theorem.

\section{Concluding remarks}

The first author has implemented the results of this paper in a SageMath package which allows for a fully automatic computation of the examples throughout this document.

The software, still under active development at the time of writing, is distributed under the GNU 
General Public License\footnote{See \url{https://www.gnu.org/licenses/gpl-3.0.txt}}, and is available at:
\vspace{-0.25em}\begin{center}\packurl\end{center}\vspace{-0.25em}

We conclude by listing some possible extensions of the results of this paper.

\begin{itemize}
   \item \textbf{Better analysis of sections:} in Examples~\ref{exm:apery2} and~\ref{exm:apery3} we obtained solutions analyzing only the first column of the matrix of operators $\cR L$. This is helpful to study solutions with a fixed kernel $K(n,k)$. However, it may happen that this approach yields no nonzero solutions. As shown in Example~\ref{ex:main}, we can extract more information if we analyze the solutions of all the columns of $\cR L$. It could be interesting to study how we can solve these systems in an automatic fashion, and what kind of information 
they can provide.

   \item \textbf{Compatibilities of derivations:} although compatibility with the derivation operator $D$ seems limited by Proposition~\ref{compat_der_roots}, there are factorial bases that are compatible with derivation operators. All the results extending compatibilities for product bases (Definition~\ref{def:pbasis}), sieved bases, and shuffled bases (Definition~\ref{def:shuffle}) can be extended for arbitrary derivation operators in a similar way as for endomorphisms (Theorem~\ref{prod} and Theorem~\ref{shuffle_E}).

   These results are already implemented in the package~\packname.
   
   \item \textbf{Other polynomial bases:} this paper has focused on factorial bases. However, we can prove Theorem~\ref{RBL} for any polynomial basis $\langle P_n(x)\rangle_{n=0}^\infty$ with $\deg P_n(x) = n$, hence we can study similar compatibility problems for orthogonal polynomial bases. 
   
   A basic implementation of these properties and bases is included in our package (see documentation for the class \texttt{OrthoBasis}).
   
   \item \textbf{Other series bases:} similarly, we can study other types of bases for $\K[[x]]$. Instead of having a basis consisting of polynomials, we can consider a formal power series basis $\langle f_n(x)\rangle_{n=0}^{\infty}$ where $f_n(x) = \sum_{k=n}^\infty c_k x^k$ with $c_n \ne 0$. In this setting, the same definition of compatibility carries over, and a corresponding version of Theorem~\ref{RBL} can be proven. 

   In particular, let $f(x) \in \K[[x]]$ have order $1$. Then if an operator $L$ is compatible with the basis $\langle f(x)^n\rangle_{n=0}^\infty$, all the solutions for $L$ can be written as a composition of a holonomic function with $f(x)$.

   A simple implementation for this type of bases is included in our software (check documentation for \texttt{OrderBasis} and \texttt{FunctionalBasis}).
\end{itemize}

\appendix

\section{Implementation}
\label{app:imp}

All the results in this paper are included in the SageMath package named \packname, which allows a fully
automated computation of all the examples throughout this document.

At the time of writing, this software is still under development (the current version is \emph{v0.3}) and has not been added to the official Sage distribution.
Readers are invited to test the functionalities included in the package and report any desired features, errors or 
comments.

The software is distributed under the GNU General Public License\footnote{See \url{https://www.gnu.org/licenses/gpl-3.0.txt}}
on the GitHub repository: 
\begin{center}\packurl\end{center}

Any Sage user can install it locally using the PyPi system included in Sage by running the command
\begin{center}\small\texttt{sage --pip install git+\packurl}\end{center}
or by cloning the repository and running \texttt{make install} in the repository folder. This process will 
install all required dependencies for a proper functionality. Once installed, the package 
is available in Sage and can be imported with the code:

\begin{lstlisting}[numbers=none]
   sage: from pseries_basis import *
\end{lstlisting}

In case the user does not want or could not install Sage locally, we offer the possibility of using it via Binder. A 
complete demo of the package with explanations of its implementation can be found at:
\begin{center}\packbinder\end{center}

All the documentation of the code can be also found at:
\begin{center}\packdoc\end{center}

\subsection{Data structures}

The package \packname~provides a class \code{FactorialBasis} to represent the factorial bases $\cB$ described 
throughout this paper (see Definition~\ref{def:factorial}). These bases have a main method \code{element} that, 
given an index $n \in \N$, returns their $n$th element.

Then, several general functionalities are included to manage the compatibilities with linear operators:
\begin{itemize}
   \item \code{set\_compatibility}: given an operator $L$ and some coefficients $\alpha_{k,j,i}$, it sets
   the compatibility of the operator $L$ with the provided coefficients (see Definitions~\ref{def:compat}
   and~\ref{def:sieved_compatible}).
   \item \code{compatibility}: given an operator $L$, it returns the compatibility coefficients 
   associated with it if it is compatible with $\cB$. 
   \item \code{recurrence}: returns the recurrence equation (or system in case of sieved bases) 
   associated with an operator $L$. This method is based on Theorem~\ref{iso}.
\end{itemize}

\subsection{Building factorial bases}

The package includes several built-in bases that can be easily obtained in the code:

\begin{itemize}
   \item \emph{Power basis}: let $\cP_{a,b} = \langle (ax+b)^n\rangle$. This basis can be built using \code{PowerBasis($a$,$b$)}
   and includes automatically the compatibility with $X$ and $D$. In particular, we can build the power basis $\cP = \langle x^n\rangle$
   that has been used in the previous sections.
   \item \emph{Falling factorial basis:} let $\cB = \langle \prod_{k=0}^{n-1} (ax + b - kc)\rangle$. This basis can be built 
   using \code{FallingBasis($a$,$b$,$c$)} and includes automatically the compatibility with $X$ and 
   $E_{c/a}: x \mapsto x + (c/a)$. This includes the falling factorial basis (when $a=1$, $b=0$ and $c=1$) and the 
   rising factorial basis (when $a=1$, $b=0$ and $c=-1$).
   \item \emph{Binomial basis:} recall $\cC_{a,b} = \langle \binom{ax+b}{n}\rangle$. These bases can be obtained
   using \code{BinomialBasis($a$,$b$)} and include automatically their compatibility with $X$ and $E$.
\end{itemize}

From these basic pieces, the user can build even further bases with the following functionality:

\begin{itemize}
   \item \emph{Scalar product:} given a basis $\cB = \langle P_n(x)\rangle$ and a hypergeometric sequence $(a_n)_n$, 
   the new basis $a\cB = \langle a_nP_n(x)\rangle$ can be computed with usual multiplication in Sage. Moreover, the 
   compatibilities of $\cB$ are automatically extended to $a\cB$.
   \item \emph{Product basis:} the product basis of $\cB_1,\ldots,\cB_m$ can be built in the code with the command
   \code{ProductBasis([$\cB_1$,\dots,$\cB_m$], ends=[$\cdot$])}, where the content of \code{ends} is a list of the 
   names of endomorphisms that the resulting product basis will be compatible with.
   \item \emph{Shuffled basis:} more generally, the user can build a shuffled basis with the command 
   \code{SievedBasis([$\cB_1$,\dots,$\cB_m$],$\mathbf{c}$,ends=[$\cdot$])}, where $\mathbf{c}$ is the cycle 
   determining how the roots of the factor bases are shuffled (see Definition~\ref{def:shuffle}).
\end{itemize}

\subsection{Building generalized binomial bases}

In Examples~\ref{exm:apery2} and~\ref{exm:apery3}, we have used a basis whose even elements were the binomial
coefficients $\binom{x+n}{2n}$. In this subsection we illustrate how we can use our package to build this 
basis automatically.

By shifting $n\mapsto (n+1)$ in $\binom{x+n}{2n}$ and taking the quotient, we obtain:
\[\frac{\binom{x+n+1}{2n+2}}{\binom{x+n}{2n}} = \frac{(x+n+1)(x-n)}{(2n+ 1)(2n+2)},\]
so at every two steps we added the roots $n$ and $(-n-1)$ to the root sequence of the basis and 
a factor of $\frac{1}{(2n+1)(2n+2)}$ to the leading coefficients.

The falling and rising factorial basis have root sequences
\[\rho_f = \langle 0,1,2,\ldots\rangle,\quad \rho_r = \langle -1,-2,-3,\ldots\rangle,\]
respectively. Hence, we can build the product of these two bases to obtain a basis with the 
desired root sequence:

\begin{lstlisting}[numbers=none]
   sage: pos_roots = FallingBasis(1,0,1)
   sage: neg_roots = FallingBasis(1,1,-1)
   sage: almost = ProductBasis([pos_roots,neg_roots])
\end{lstlisting}

In this piece of code, the object \code{almost} contains a product basis that guarantees
the desired root sequence. This is not yet the basis we want because the leading coefficient sequence is all 1:
\begin{lstlisting}[numbers=none]
   sage: [almost.cn()(i) for i in range(10)]
   [1, 1, 1, 1, 1, 1, 1, 1, 1, 1]
\end{lstlisting}

And if we look at the bases used in Examples~\ref{exm:apery2} and~\ref{exm:apery3}, the leading coefficients are 
precisely $1/n!$. We can build this also in the code since $1/n!$ is a hypergeometric sequence:
\begin{lstlisting}[numbers=none]
   sage: basis = (1/factorial(n))*almost
\end{lstlisting}

This process can be generalized to any desired binomial coefficient of the shape
\[\binom{ax+bn+c}{mn+r}\]

This is automatized in the method \code{GeneralizedBinomial} that receives the constant parameters $a$, $b$, $c$, $m$ and 
$r$ and returns a basis whose $(nm)$-th elements are precisely the binomial coefficients shown above. Moreover, the 
compatibilities with $X$ and $E$ are automatically computed whenever they are possible (i.e, whenever $r=0$).

\subsection{Revisiting Example~\ref{exm:apery2}}

Now, we are going to show how to use the package to reproduce Example~\ref{exm:apery2}. We are interested in 
studying compatibility of the linear operator
\[L = (n+2)^2E^2 - (11n^2 + 33n + 25)E - (n+1)^2,\]
with respect to the kernel:
\[K(n,k) = \binom{n}{k}\binom{n+k}{2k}.\]

This kernel is built as a product of two simpler binomial coefficients that can be created with the code:
\begin{lstlisting}[numbers=none]
  sage: b1 = BinomialBasis(1,0)
  sage: b2 = GeneralizedBinomial(1,1,0,2,0)
\end{lstlisting}

As we saw in Example~\ref{exm:apery2}, we are interested in a basis that has $K(n,k)$ as some of its elements.
This can be achieved with a shuffled basis:

\begin{lstlisting}[numbers=none]
  sage: A2 = SievedBasis([b1,b2],[1,0,1], ends=['E'])
\end{lstlisting}

Now that we have built the basis of $\K[x]$, we need to build the linear operator associated to $L$. For doing so, we 
are going to use the package \code{ore\_algebra}\footnote{\url{https://github.com/mkauers/ore_algebra}} developed
by M.~Kauers and M.~Mezzarobba that provides functionality to represent such linear operators:

\begin{lstlisting}[numbers=none]
  sage: OE.<E> = OreAlgebra(QQ[x], ('E', 
                              lambda p : p(x=x+1), 
                              lambda p : 0))
  sage: L = (x+2)^2*E^2 - 
            (11*x^2 + 33*x + 25)*E - 
            (x+1)^2
\end{lstlisting}

And now we follow the process described in Example~\ref{exm:apery2}: we build the recurrence matrix, 
take the first column and compute a greatest common right divisor of its elements:

\begin{lstlisting}[numbers=none]
  sage: recurrence_matrix = A2.recurrence(L)
  sage: first_column = 
    [A2.remove_Sni(recurrence_matrix[j,0]) 
    for j in range(recurrence_matrix.nrows())]
  sage: gcrd = first_column[0].gcrd(*first_column[1:])
  sage: gcrd
  (n+1)*Sn - 4n - 2
\end{lstlisting}
Which is exactly the recurrence we obtained in~\eqref{eq:apery2:gcrd}.

\subsection{Revisiting Example~\ref{exm:apery3}}

In Example~\ref{exm:apery3} we studied compatibility of the linear operator
\[L = (n + 2)^3 E^2 - (2n+3)(17 n^2 + 51 n + 39) E + (n + 1)^3,\]
with respect to the kernel:
\[K(n,k) = \binom{n+k}{2k}^2.\]

This kernel is the square of a simple binomial coefficient that can be created with the code:
\begin{lstlisting}[numbers=none]
  sage: b2 = GeneralizedBinomial(1,1,0,2,0)
\end{lstlisting}

 This basis contains $\binom{x+n}{2n}$ as its even positions. Hence, to obtain the basis
 described in Example~\ref{exm:apery3}, we only need to build the \texttt{ProductBasis} of \texttt{b2}
 with itself:

\begin{lstlisting}[numbers=none]
  sage: A3 = ProductBasis([b2,b2],ends=['E'])
\end{lstlisting}

In a similar way as we did in Example~\ref{exm:apery2}, we now build the linear operator using 
\texttt{ore\_algebra}, then obtain the recurrence matrix and consider the greatest common right divisor
of the elements of its first column:

\begin{lstlisting}[numbers=none]
  sage: OE.<E> = OreAlgebra(QQ[x], ('E', 
   lambda p : p(x=x+1), 
   lambda p : 0))
  sage: L = (x+2)^3*E^2 - 
   (2*x + 3)*(17*x^2 + 51*x + 39)*E + 
   (x+1)^3
  sage: recurrence_matrix = A3.recurrence(L)
  sage: first_column = 
    [A3.remove_Sni(recurrence_matrix[j,0]) 
    for j in range(recurrence_matrix.nrows())]
  sage: gcrd = first_column[0].gcrd(*first_column[1:])
  sage: gcrd
  (n^2 + 2*n + 1)*Sn - 16*n^2 - 16*n - 4
\end{lstlisting}
Which is exactly the recurrence we obtained in~\eqref{eq:apery3:gcrd}.

\section{More examples}
\label{app:examples}
In this section we include additional examples that are not in the original paper. These examples showcase how to use the results on this paper, and illustrate (see Examples~\ref{exm:double1} and \ref{exm:double2})
how we can obtain definite-sum solutions containing not just one, but several nested definite sums.

All these examples (and more) can be found and tested in the repository\footnote{\url{https://github.com/Antonio-JP/pseries_basis/blob/master/notebooks/paper_examples_appB.ipynb}}. We also offer a 
binder notebook to try these examples out without installing Sage or the package:
\begin{center}\packappb\end{center}

\begin{example}[Binomial transform of the Catalan numbers]\label{exm:bincat}
   Let us start with the sequence $(e_n)_n$ defined by $e_0 = 1$, $e_1 = 2$ and $L e = 0$ where:
   \[L = \left(x + 3\right) E^{2} -2 \left(3 x + 5\right) E + 5 (x + 1).\]
   According to the OEIS database, this sequence\footnote{\url{https://oeis.org/A007317}} is the binomial transform of the Catalan numbers, i.e., 
   \[e_n = \sum_{k=0}^n c_k \binom{n}{k},\]
   where $(c_n)_n$ is the sequence of Catalan numbers\footnote{\url{https://oeis.org/A000108}}. The methods of this paper are a great tool for proving automatically this type of 
   identities. For doing so, we compute a new sequence $(b_n)_n$ such that we know $e_n = \sum_{k=0}^n b_k\binom{n}{k}$. This sequence will be annihilated by $\cR[\cC](L)$ and will have
   as initial conditions
   \[b_0 = 1, b_1 = 1, b_2 = 2, b_3 = 5.\]

   Computing $\cR[\cC](L)$ yields the operator:
   \[\cR[\cC](L) = \left(n + 3\right) \mathit{S_n}^{2} - \left(3 n + 4\right) \mathit{S_n} - 2(2 n +1).\]
   By closure properties of P-recursive sequences it is easy to show that $b_n = c_n$ for all $n\in \N$. 
\end{example}

\begin{example}[Franel numbers]\label{exm:franel}
   Franel numbers\footnote{See \url{https://oeis.org/A000172}} $(f_n)_n$ satisfy $f_0 = 1$, $f_1 = 2$ and $L f = 0$ where $L = (n+2) ^2 E^2 - (7n^2 +21n+16)E - 8(n+1)^2$.

   It is known that Franel numbers are the sum of the cubes of binomial coefficients: $f(n) = \sum_{k=0}^n \binom{n}{k}^3$. We can check this identity using our methods
   with the product basis $\cC_{(1,1,1),(0,0,0)}$. In this case, the compatibility of $\cC_{(1,1,1),(0,0,0)}$ with $E$ and $X$ can be written
   with $3\times 3$ matrices:
   \[\cR[](X) = \begin{pmatrix}n & 0 & nS_n^{-1}\\ (n+1) & n & 0\\ 0 & (n+1) & n\end{pmatrix},\]
   \[\cR[](E) = \begin{pmatrix}S_n +1 & \frac{3n+1}{n+1} & \frac{3n^2+3n+1}{(n+1)^2}\\ 3S_n & \frac{n+1}{n+2}S_n + 1 & \frac{3n+2}{n+1}\\3S_n & \frac{3n+3}{n+2}S_n & \frac{(n+1)^2}{(n+2)^2} S_n +1\end{pmatrix}.\]

   The associated matrix for the operator $L$ that defines the Franel numbers is then a $3\times 3$ matrix. We have included in the GitHub repository a folder with the description of 
   each of its elements\footnote{\url{https://github.com/Antonio-JP/pseries_basis/tree/master/notebooks/example57}}. But since we want to see solutions of the shape 
   \[f(x) = \sum_{n} c_n \binom{x}{n}^3,\]
   we have to do as we did for Examples~\ref{exm:apery2} and~\ref{exm:apery3}, and consider the greatest common right divisor of the elements of the first column of the matrix which are:
   \begin{align*}
      L_{0,0} ={}& (n+3)^2 S_n^3 + \frac{158n^4 + 686n^3 + 1088n^2 + 756n + 199}{(n+1)^2} S_n^2 + \\ 
                & \frac{62n^4 - 8n^3 - 334n^2 - 402n - 141}{(n+1)^2}S_n - (221n^2 + 244n + 67)\\
      L_{1,0} ={}& \frac{(n+3)(11n^2 + 42n + 37)}{n+2} S_n^3 + \\
                & \frac{274n^4 + 1598n^3 + 3414n^2 + 3166n + 1073}{(n+1)(n+2)} S_n^2 - \\ 
                & \frac{200n^3 + 883n^2 + 1217n + 531}{(n+1)}S_n - (85n + 61)(n+1)\\
      L_{2,0} ={}& \frac{55n^4 + 420n^3 + 1168n^2 + 1398n + 607}{(n+2)^2} S_n^3 + \\
                & \frac{266n^4 + 1806n^3 + 4568n^2 + 5106n + 2129}{(n+2)^2} S_n^2 - \\ 
                & (307n^2 + 914n +670)S_n - 14(n+1)^2
   \end{align*}

   Computing the \emph{greatest common right divisor} for these operators yields the operator $S_n - 1$, which is only satisfied by constant sequences. A simple
   computation guarantees that if $f(x) = \sum_{n} c_n \binom{x}{n}^3$ for $f(x)$ the solution yielding the Franel numbers, then $c_0 = 1$ meaning that $c_n = 1$
   for all $n$ proving the desired identity:
   \[f_n = \sum_{k=0}^n\binom{n}{k}^3\quad\text{for all $n\in \N$}.\]
\end{example}

\begin{example}[First double binomial sum]\label{exm:double1}
   Let us consider now the sequence $(d_n)_n$ where $d_n$ is the \emph{sum over all Dyck paths of semilength $n$ of the arithmetic 
   mean of the $x$ and $y$ coordinates}\footnote{\url{https://oeis.org/A258431}}. It is known that this sequence satisfies the following
   recurrence
   \[(n-1)d_n = (8n-10)d_{n-1}-(16n-24)d_{n-2}\text{ for $n > 2$},\]
   and has as its first terms $d_0 = 0$, $d_1 = 1$ and $d_2 = 5$. Let $d(x)$ be a function such that $d(n) = d_n$ for all $n \in \N$. Then 
   it is annihilated by the following linear operator:
   \begin{equation}\label{equ:double}
      L = \left(x + 2\right) E^{3} - \left(8 x + 14\right) E^{2} + \left(16 x + 24\right) E.
   \end{equation}

   We want to find an explicit formula for this sequence, so we try to compute a sum with respect to the binomial basis $\cC$,
   i.e., we write $d(x) = \sum_{k} c_k \binom{x}{k}$. Computing the operator $\cR[\cC](L)$ yields:
   \[\left(n + 3\right) S_n^{4} - \left(4 n + 12\right) S_n^{3} - 2 n S_n^{2} + 12\left( n + 2\right) S_n + 9 (n + 1).\]

   Hence, the sequence $(c_n)_n$ defined by $\cR[\cC](L)$ with initial terms $0,1,3,11,36,\ldots$ allows us to write:
   \[d_n = \sum_{k=0}^n c_k\binom{n}{k}.\]

   If we look at OEIS for this sequence $(c_n)_n$, we do not find anything. So we compute now a new sequence $(b_n)_n$ for which 
   we can write $c_n = \sum_{l=0}^n b_k\binom{n}{k}$. To this end, we compute the operator:
   \[\cR[\cC](\cR[\cC](L)) = \left(n + 3\right) \mathit{S_n}^{4} + n \mathit{S_n}^{3} -2 \left(4 n + 9\right) \mathit{S_n}^{2} - 8 n \mathit{S_n} + 16 n \mathit{S_n}^{-1} + 8(2n + 3).\]

   Using the results of this paper, we know that the sequence $(b_n)_n$ with initial terms $0,1,1,5,6$ is annihilated by $\cR[\cC](\cR[\cC](L))$.
   This sequence still does not show up in OEIS. However, we observe that the sequence $(b_n)_n$ is the interlacing of two simpler sequences. More specifically, the 
   odd terms look like the original sequence $(d_n)_n$. Hence, if we put everything together, we can prove by using closure properties that, if $(d_n)_n$ are defined as 
   above and $(b_n)_n$ are defined by the formula:
   \[d_n = \sum_{k=0}^l \sum_{l=0}^k b_l \binom{k}{l}\binom{k}{n},\]
   then for all $n\in \N$, we have $b_{2n-1} = d_n$.
\end{example}

\begin{example}[Second double binomial sum]\label{exm:double2}
   For this example we are going to consider the sequence:
   \[a_n = \sum_{k=0}^n\binom{2n}{k}.\]
   This sequence\footnote{\url{https://oeis.org/A032443}} is half of the sum of the binomial coefficients having even upper argument. Let $a(x)$ be a function such that $a(n) = a_n$. We can check that this 
   function is annihilated by the following recurrence operator:
   \[L = \left(x + 2\right) E^{2} - 2\left(4 x +5\right) E + 8(2 x + 1).\]

   We proceed now similarly to Example~\ref{exm:double1}. Let $(c_n)_n$ be a sequence with $a_n = \sum_{k=0}^n c_k\binom{n}{k}$. Then $(c_n)_n$ is annihilated 
   by $\cR[\cC](L)$, which we can compute:
   \[\cR[\cC](L) = \left(n + 2\right) \mathit{S_n}^{2} - \left(5 n + 6\right) \mathit{S_n} + 3 n + 9 n \mathit{S_n}^{-1}.\]

   This recurrence operator involves the inverse shift $S_n^{-1}$. In order to apply again the recurrence compatibility with the binomial basis $\cC$, we need to remove this 
   inverse shift. For doing so we simply multiply $\cR[\cC](L)$ by $S_n$ from the left. Our sequence $(c_n)_n$ is still annihilated by $S_n\cR[\cC](L)$. 
   
   The sequence $(c_n)_n$ can be found in OEIS as A027914, defined as the sum of the first half of trinomial coefficients. Let us now consider the sequence $(b_n)_n$ defined
   again as $c_n = \sum_{k=0}^n b_k\binom{n}{k}$. The sequence $(b_n)_n$ is annihilated by $\cR[\cC](S_n\cR[\cC](L))$, which we can compute:
   \[\cR[\cC](S_n\cR[\cC](L)) = \left(n + 3\right) \mathit{S_n}^{3} - \left(n + 2\right) \mathit{S_n}^{2} - 2\left(3n + 5\right) \mathit{S_n} + 4(n + 1) + 8 n \mathit{S_n}^{-1}.\]

   The sequence $(b_n)_n$, defined by $\cR[\cC](\cR[\cC](L))$ and the initial terms 
   \[b_0 = 1,\quad b_1 = 1,\quad b_2 = 3,\quad b_3 = 4,\quad b_4 = 11,\]
   appears in OEIS as the sequence A027306 (which can be checked automatically using closure properties of D-finite sequences). This sequence has a closed form formula:
   \[b_n = 2^{n-1} + \left(\frac{1 + (-1)^n}{4}\right)\binom{n}{n/2}.\]

   Putting everything together, this process has proved the following identity:
   \[\sum_{k=0}^n\binom{2n}{k} = \sum_{k=0}^n\sum_{l=0}^k \left(2^{l-1} + \left(\frac{1 + (-1)^l}{4}\right)\right)\binom{l}{l/2} \binom{k}{l}\binom{n}{k}.\]
\end{example}
         
\begin{example}[Third double binomial sum]\label{exm:artificial}
   We can illustrate the recursive use of our methods on the linear recurrence equation $Ly = 0$ where $L \in \Q[x]\langle E\rangle$ is the recurrence operator of order 5 defined
   by
   \begin{align*}
      L\ &=\ -64 (1 + x) (2 + x) (3 + x) (-151 - 39 x + 8 x^2 + 2 x^3)\\
      &+ 16 (2 + x) (3 + x) (-3867 - 2400 x - 182 x^2 + 108 x^3 + 16 x^4) E\\
      &-4 (3 + x) (-48214 - 42707 x - 10472 x^2 + 530 x^3 + 500 x^4 + 48 x^5) E^2\\
      &+ 2 (-163088 - 179069 x - 66637 x^2 - 6360 x^3 + 1822 x^4 + 480 x^5 + 32 x^6) E^3\\
      &- (-61566 - 62939 x - 21344 x^2 - 1644 x^3 + 580 x^4 + 132 x^5 + 8 x^6) E^4\\
      &+ (5 + x) (-106 - 49 x + 2 x^2 + 2 x^3) E^5.
   \end{align*}

   If we look for solutions $y(x)$ such that $L y = 0$ we will not find any hypergeometric solutions. Then we can try to use the methods of this paper to 
   find a definite-sum solution for this recurrence. We start by taking the binomial basis $\cC$ that we have used throughout the paper. If we write the solutions 
   $y(x)$ in the binomial basis $y(x) = \sum_{n=0}^\infty z_n \binom{x}{n}$, then we know that the sequence $(z_n)_n$ is annihilated by the recurrence $\cR[\cC](L)$:

   \begin{align*}
      \cR[\cC](L) &= -8 (n-5) (n-4) (n-3) (n-2) (n-1) n\, S^{-6}\\
      &- 4 (n-4) (n-3) (n-2) (n-1) n (15 + 4 n)S^{-5}\\
      &+ 2 (n-3) (n-2) (n-1) n (17 - 114 n + 12 n^2)S^{-4}\\
      &+ 2 (n-2) (n-1) n (451 - 493 n + 32 n^2 + 32 n^3)S^{-3}\\
      &- (n-1) n (-941 + 2804 n + 692 n^2 - 576 n^3 + 16 n^4)S^{-2}\\
      &- n (4128 + 8413 n + 1816 n^2 - 1772 n^3 + 152 n^4 + 96 n^5)S^{-1}\\
      &+ 180 + 10988 n + 19519 n^2 + 8440 n^3 - 752 n^4 - 680 n^5 - 16 n^6\\
      &+ (50990 + 66145 n + 22205 n^2 - 1808 n^3 - 1000 n^4 + 288 n^5 + 64 n^6) S\\
      &+ (-35864 - 77301 n - 45009 n^2 - 6616 n^3 + 1596 n^4 + 472 n^5 + 24 n^6) S^2\\
      &- (85212 + 67646 n + 12843 n^2 - 1308 n^3 - 236 n^4 + 108 n^5 + 16 n^6) S^3\\
      &- (-58916 - 60882 n - 21043 n^2 - 1728 n^3 + 562 n^4 + 132 n^5 + 8 n^6) S^4\\
      &+ (5 + n) (-106 - 49 n + 2 n^2 + 2 n^3) S^5.
   \end{align*}

   In particular, the sequence $(z_n)_n$ will be annihilated by $S^6 \cR[\cC](L)$. This operator
   only has forward shifts and can be considered as an element of $\Q[x]\langle E\rangle$ (as
   the operator $L$), by mapping $S \mapsto E$ and $n \mapsto x$. Now, we consider a function
   $z(x)$ with $z(n) = z_n$ that is annihilated by the following operator:

   \begin{align*}
      M &= -8 (x+1) (x+2) (x+3) (x+4) (x+5) (x+6)\\
      & -4 (x+2) (x+3) (x+4) (x+5) (x+6) (39 + 4 x)E\\
      & +2 (x+3) (x+4) (x+5) (x+6) (-235 + 30 x + 12 x^2) E^2\\
      & + 2 (x+4) (x+5) (x+6) (5557 + 3347 x + 608 x^2 + 32 x^3) E^3\\
      & -  (x+5) (x+6) (-62885 - 37276 x - 6220 x^2 - 192 x^3 + 16 x^4) E^4\\
      & - (x+6) (680718 + 592237 x + 210112 x^2 + 36436 x^3 + 3032 x^4 + 96 x^5) E^5\\
      & - (4416936 + 4645888 x + 1770833 x^2 + 323528 x^3 + 29792 x^4 +  1256 x^5 + 16 x^6)  E^6 \\
      & + (4786184 + 4125565 x + 1639901 x^2 + 354352 x^3 + 42200 x^4 + 2592 x^5 + 64 x^6) E^7\\
      & + (3309382 + 4225311 x + 1666719 x^2 + 305288 x^3 + 28716 x^4 + 1336 x^5 + 24 x^6) E^8\\
      & - (1951356 + 1322930 x + 482643 x^2 + 101028 x^3 + 11644 x^4 + 684 x^5 + 16 x^6) E^9\\
      & - (573028 + 1214154 x + 509885 x^2 + 93840 x^3 + 8842 x^4 + 420 x^5 + 8 x^6)  E^{10}\\
      & + (11 + x) (104 + 191 x + 38 x^2 + 2 x^3) E^{11},
   \end{align*}
   which is the same operator as $S^6\cR[\cC](L)$ but written as an element of $\Q[x]\langle E\rangle$.
   Similarly to what we did with $y(x)$, we can write the function $z(x) =\sum_{k = 0}^\infty w_k \binom{x}{k}$ and 
   then the sequence $(w_n)_n$ is annihilated by the operator $\cR[\cC](M)$:

   \begin{align*}
      \cR[\cC](M) &=256 (n-2)^2 (n-1) n (2 n-3) S^{-3}\\
      & -128 (n-1) n \left(4 n^4-66 n^3+124 n^2-77 n+14\right) S^{-2}\\
      & -32 n \left(144 n^5-1036 n^4-954 n^3+52 n^2+n+32\right) S^{-1}\\
      & -16 \left(1176 n^6-260 n^5-18480 n^4-42580 n^3-37988 n^2-12369 n-252\right)\\
      & -16 \left(2880 n^6+16930 n^5+9973 n^4-120405 n^3-307141 n^2-257421 n-54488\right) S\\
      & -8 \left(9420 n^6+109214 n^5+441622 n^4+610729 n^3-368468 n^2-1447147 n-656803\right) S^2\\
      & -2 \left(43344 n^6+740620 n^5+4880834 n^4+15304796 n^3+22173471 n^2+11018808 n-25704\right) S^3\\
      & -\left(71912 n^6+1612772 n^5+14331384 n^4+63562068 n^3+144253438 n^2+150787607 n+51445066\right) S^4\\
      & -\left(43344 n^6+1198400 n^5+13254962 n^4+74139652 n^3+216230077 n^2+297858761 n+134496052\right) S^5\\
      & -\left(18840 n^6+617420 n^5+8122956 n^4+54243220 n^3+189442880 n^2+311777199 n+161972208\right) S^6\\
      & -\left(5760 n^6+217784 n^5+3308558 n^4+25513428 n^3+102653179 n^2+192498561 n+107373416\right) S^7\\
      & -\left(1176 n^6+50300 n^5+863968 n^4+7519420 n^3+33976774 n^2+70410893 n+39850322\right) S^8\\
      & -\left(144 n^6+6864 n^5+131182 n^4+1266684 n^3+6307771 n^2+14115123 n+7618716\right) S^9\\
      & -\left(8 n^6+420 n^5+8812 n^4+93012 n^3+502140 n^2+1188245 n+560444\right) S^{10}\\
      & +(n+11) \left(2 n^3+38 n^2+191 n+104\right) S^{11}.
   \end{align*}
   If we now clean the inverse shifts as we did for $\cR[\cC](L)$ then we obtain the operator of order 14 defined by $N = S^3 \cR[\cC](M)$. 
   We again transform the operator $N$ to the operator ring $\Q[x]\langle E\rangle$ by mapping $S \mapsto E$ and $n \mapsto x$:

   \begin{align*}
      N &= 256 (x+1)^2 (x+2) (x+3) (2 x+3)\\
      & -128 (x+2) (x+3) \left(4 x^4-18 x^3-254 x^2-683 x-559\right) E\\
      & -32 (x+3) \left(144 x^5+1124 x^4-426 x^3-25598 x^2-79013 x-74179\right) E^2\\
      & -16 \left(1176 x^6+20908 x^5+136380 x^4+347300 x^3-60488 x^2-1776489x - 2231667 \right) E^3\\
      & -16 \left(2880 x^6+68770 x^5+652723 x^4+3078171 x^3+7218056 x^2+6781572 x+179368 \right) E^4\\
      & -8 \left(9420 x^6+278774 x^5+3351532 x^4+20826253 x^3+69908761 x^2+118492934 x+77352791 \right) E^5\\
      & -2 \left(43344 x^6+1520812 x^5+21841574 x^4+163936364 x^3+676112031 x^2+1447565850
         x+1252737441\right) E^6 \\
      & -\left(71912 x^6+2907188 x^5+48231084 x^4+419520636 x^3+2013028306 x^2+5038293899 x+5123434213\right) E^7\\
      & -\left(43344 x^6+1978592 x^5+37082402 x^4+364460956 x^3+1975485853 x^2+5577093275 x+6372374530\right) E^8\\
      & -\left(18840 x^6+956540 x^5+19927656 x^4+217460092 x^3+1305865484 x^2+4067804487 x+5088583521\right) E^9\\
      & - \left(5760 x^6+321464 x^5+7352918 x^4+87927084 x^3+576736243 x^2+1951205055 x+ 2622724016 \right) E^{10}\\
      & -\left(1176 x^6+71468 x^5+1777228 x^4+23049076 x^3+163315666 x^2+592690529 x+ 842959919 \right) E^{11}\\
      & -\left(144 x^6+9456 x^5+253582 x^4+3536388 x^3+26819995 x^2+103319745 x+153333162\right) E^{12}\\
      & -\left(8 x^6+564 x^5+16192 x^4+240876 x^3+1938216 x^2+7845869
         x+11977427\right)E^{13}\\
      & + (x+14) \left(2 x^3+56 x^2+473 x+1073\right) E^{14}
   \end{align*}

   Now we can use any software to find explicit solutions to $N v = 0$. For example, we find two linearly independent hypergeometric solutions:
   \begin{align*}
      v^{(1)}(n) &= (2n+1)!,\\
      v^{(2)}(n) &= \frac{1}{n!}.
   \end{align*}   

   Hence, we can write two linearly independent solutions of the original operator $L$ by unrolling the identities we have found:
   \begin{align*}
      y^{(1)}(n) &= \sum_{k=0}^n \sum_{j=0}^k \binom{n}{k}\binom{k}{j}(2j+1)!,\\
      y^{(2)}(n) &= \sum_{k=0}^n \sum_{j=0}^k \binom{n}{k}\binom{k}{j}\frac{1}{k!},
   \end{align*}
   both satisfy $L y = 0$.

   The operators $L$, $M$ and $N$ can be found in the repository together with the code\footnote{\url{https://github.com/Antonio-JP/pseries_basis/tree/master/notebooks/example60}}.
\end{example}
\section{The matrix elements of  \texorpdfstring{$[\cR L]$}{[RL]} from Example~\ref{ex:main}}
\label{app:ex:main}

Here we present explicit formulas for the operators $L_{0,0},L_{0,1},L_{1,0}$ and $L_{1,1}$ from 
Example~\ref{ex:main}. The computations for obtaining these operators are explained in detail in the aforementioned 
Example, but can also be automatically computed using the software \packname. They can also be found (ready to 
be used in SageMath) in the repository\footnote{\url{https://github.com/Antonio-JP/pseries_basis/tree/master/notebooks/example44}}.

\begin{align*}
L_{0,0}\ &=\ (k+8) (27034107689 k+247037440535)\, S^7\\
&-\frac{1}{(k+1) (k+2) (k+3) (k+4) (k+5)
   (k+6)}(54068215378 k^9\\
&-315669611138 k^8-45148617745347 k^7-782696842132919 k^6\\
&-6454240392445055 k^5-30050534179653883 k^4-82215116461457480 k^3\\
&-129113197043173300 k^2-105757314240946896 k-34247146225582080) \,S^6\\
&+\frac{1}{(k+1) (k+2) (k+3) (k+4) (k+5)}(27034107689 k^9-3508146051312 k^8\\
&-127964289486598 k^7-1741174847222631 k^6-12524498803569684 k^5\\
&-53047967564919031 k^4-136503346354387959 k^3-209045777928727562 k^2\\
&-173958661328786224 k-59682736706956320) \,S^5\\
&+\frac{1}{(k+1) (k+2) (k+3) (k+4)}(1972211122835 k^8+62134267567378 k^7\\
&+616718084410852 k^6+2619141590805683 k^5+4315508250526315 k^4\\
&-1167669149632785 k^3-9620009176334670 k^2-4014526382135216 k\\
&+3400385599899936) \,S^4\\
&-\frac{1}{(k+1) (k+2) (k+3)}(2972566483581 k^7-148275885358425 k^6\\
&-2589937042152480 k^5-16499978058431541 k^4-52671318556586357 k^3\\
&-88255097542772662 k^2-71905587088529204
   k-21524025761438520) \,S^3\\
&-\frac{1}{(k+1) (k+2)}(64025119688979 k^6+916316298831859 k^5\\
&+5515823411381379 k^4+17451451407071553 k^3+30016470047039710 k^2\\
&+26136807998134436 k+8854550588334008) \,S^2\\
&+\frac{4}{k+1}(12511390805301 k^5+48661327183573 k^4-74830042870409 k^3\\
&-512087325174801 k^2-633098967293677k-198345093160056) \,S\\
&+4 (34604693659372 k^4+175550020109206 k^3+291507102636319 k^2\\
&+199874021738859 k+49640119659704)\\
&+8 k^2 (2263487310112 k^2+6642551248868 k+2276852470297)\,S^{-1}\\
&-413236428752 (k-1)^2 k^2 \,S^{-2},
\end{align*}

\begin{align*}
L_{0,1}\ &=\ (432545723024 k^2+5219471638609 k+13834096669960) \,S^6\\
&-\frac{1}{(k+1) (k+2) (k+3) (k+4) (k+5)}(811023230670 k^8\\
&+19046230918120 k^7+165247796584595 k^6+627022054492313 k^5\\
&+729851398238169 k^4-1509733892578497 k^3-4308735593846114 k^2\\
&-2212258070133576 k+614777382717120) \,S^5\\
&+\frac{1}{(k+1) (k+2) (k+3) (k+4)(k+5)}(378477507646 k^9\\
&+2622949450120 k^8-178977301246832 k^7-3536993458109041 k^6\\
&-28948803267991653 k^5-128904824467220745 k^4-332824849676929377k^3\\
&-491118612333672390 k^2-377639831328665120 k-114882404752528800) \,S^4\\
&+\frac{1}{(k+1) (k+2) (k+3)(k+4)}(3985703076428 k^8+193788717050920 k^7\\
&+2713882670024520 k^6+18619321752485251 k^5+73349659346823626 k^4\\
&+174605063422565737 k^3+248149371724462126 k^2+193178386633211432 k\\
&+62908960035990144) \,S^3\\
&-\frac{2}{(k+1) (k+2)(k+3)} (15771010694581 k^7+181528774532964 k^6\\
&+745276353701544 k^5+977855992237626 k^4-1343180715641631 k^3\\
&-4587333163883826 k^2-2904023672777446 k+307621359311628) \,S^2\\
&-\frac{4}{(k+1)(k+2)}(10843329249882 k^6+197606913503225 k^5\\
&+1312326808327958 k^4+4181515523826039 k^3+6805531587905072 k^2\\
&+5342072311504845 k+1547863353158842) \,S\\
&+\frac{8}{k+1} (17316549266881 k^5+108087475519889 k^4+241637933180241 k^3\\
&+250891662715092 k^2+128154251184293 k+28521786904807)\\
&+16 k (4536557506826 k^3+16485306406473 k^2+14750499433633 k\\
&+4048952022402)\,S^{-1}\\
&+384 (k-1) k^2 (1547923253 k+13567625316) \,S^{-2},
\end{align*}

\begin{align*}
L_{1,0}\ &=\ \frac{1}{(k+2) (k+3) (k+4) (k+5) (k+6) (k+7)} (432545723024 k^8\\
&+17547024744793 k^7+300772935395324 k^6+2834950240712954 k^5\\
&+16000429195865408 k^4+55071479995635089 k^3+112099510188633348 k^2\\
&+122247182298335548 k+53987648288898960)\, S^7\\
&-\frac{1}{(k+2) (k+3) (k+4) (k+5) (k+6)}(811023230670 k^8\\
&+24723393532810 k^7+292833072628835 k^6+1626762750499999 k^5\\
&+3369314689609239 k^4-6528259273082053 k^3-46958468605880528 k^2\\
&-85444980480836572 k-52432761655513872) \,S^6\\
&+\frac{1}{(k+2) (k+3) (k+4) (k+5)}(378477507646 k^8+3190665711589 k^7\\
&-176165113042889 k^6-3680820024183060 k^5-30578022192058416 k^4\\
&-133391475526561039 k^3-319542866474066205 k^2-394131699873504978 k\\
&-193003564034906648) \,S^5\\
&+\frac{1}{(k+2) (k+3) (k+4)}(3985703076428 k^7+199767271665562 k^6\\
&+2926376432119304 k^5+20553495786751855 k^4+79525367763859646 k^3\\
&+173129687209637083 k^2+197180830938857338 k+90043789227857560) \,S^4\\
&-\frac{1}{(k+2) (k+3)}(31542021389162 k^6+410370581149671 k^5\\
&+1889632865572464 k^4+2742859006263721 k^3-4455809171785822 k^2\\
&-16973330513822344 k-13162727503867212) \,S^3\\
&-\frac{4}{k+2}(10843329249882 k^5+213871907378048 k^4+1533963953805732 k^3\\
&+5087511194624529 k^2+7854660846698885 k+4507808585185441) \,S^2\\
&+4\,(34633098533762 k^4+268124598840421 k^3+719049847857749 k^2\\
&+787188237460817 k+289840947961864) \,S\\
&+8 (k+1) (9073115013652 k^3+46580285333424 k^2+68719863652441 k\\
&+31063488457919)\\
&+192 k^2 (k+1) (3095846506 k+28683173885)\,S^{-1},
\end{align*}

\begin{align*}
L_{1,1}\ &=\ (k+1) (27034107689 k+247037440535)\, S^7\\
&+\frac{1}{(k+2) (k+3) (k+4) (k+5) (k+6)} (54068215378 k^8-694147118784 k^7\\
&-55428688219699 k^6-878006017531531 k^5-6681392431254415 k^4\\
&-28097754885306673 k^3-66148844332414088 k^2-80650521283582156 k\\
&-38958998544954000)\, S^6\\
&-\frac{1}{(k+2) (k+3) (k+4) (k+5) (k+6)}(27034107689 k^9-3508146051312 k^8\\
&-152565327483588 k^7-2467626608265457 k^6-21224638658568224 k^5\\
&-107972466474419135 k^4-333959033455101617 k^3-612421038700564968 k^2\\
&-605231342224899340 k-243934827298760208) \,S^5\\
&+\frac{1}{(k+2) (k+3) (k+4) (k+5)}(1972211122835 k^8+71974682766174 k^7\\
&+875773006381164 k^6+4554430750938027 k^5+6786812575418620 k^4\\
&-31583019684887547 k^3-157161970509807851 k^2-257888323765433574 k\\
&-149105098468738984) \,S^4\\
&-\frac{1}{(k+2) (k+3) (k+4)}(2972566483581 k^7-149184063635101 k^6\\
&-3473324917413492 k^5-29111963906335209 k^4-124841735272806629 k^3\\
&-292268791433903686 k^2-353135255931287712 k-170560133319385632) \,S^3\\
&-\frac{1}{(k+2) (k+3)}(64025119688979 k^6+1150730119088011 k^5\\
&+8609373148451587 k^4+33951954000293401 k^3+73784589829185334 k^2\\
&+83087956017304548 k+37513154389125452) \,S^2\\
&+\frac{4}{k+2}(12511390805301 k^5+68878950332595 k^4-68222690798060 k^3\\
&-1049675078236094 k^2-2149789355833503 k-1346833666634436) \,S\\
&+4\,(34604693659372 k^4+235686292414298 k^3+553354051695523 k^2\\
&+545211853981501 k+192577819165746)\\
&+8 k (k+1)(2263487310112 k^2+8868888400908 k+5810514296121)\,S^{-1}.
\end{align*}

\section*{Acknowledgements}

The first author acknowledges financial support from the strategic program ``Innovatives O\"O-2010 plus'' of
the Ober\"osterreich region, the Austrian Science Fund (FWF) under the grant SFB F50-07, and the Paris Ile-de-France region.

The second author acknowledges financial support from the Slovenian Research Agency (research core funding No.\ P1-0294). The paper was started in 2017 while he was attending the thematic programme ``Algorithmic and Enumerative Combinatorics'' at the Erwin Schr\"odinger International Institute for Mathematics and Physics in Vienna, Austria. He thanks the Institute for its support and warm hospitality.

Both authors thank the anonymous reviewers for their thorough and well-informed reviews which helped to improve this paper significantly.
 
\renewcommand\bibname{References}

\end{document}